\documentclass[a4paper, onecolumn,accepted=2023-12-27]{quantumarticle}
\pdfoutput=1

\usepackage[utf8]{inputenc}
\usepackage{fullpage}
\usepackage{todonotes}
\usepackage{amssymb}
\usepackage{amsmath}
\usepackage{amsthm}
\usepackage{color}
\usepackage{verbatim}
\usepackage{dsfont}
\usepackage{float}
\usepackage{tikz}
\usepackage{fixmath}
\usepackage{latexsym}
\usepackage{enumerate}
\usepackage[boxed]{algorithm2e}
\usepackage[numbers,sort&compress]{natbib}

\usetikzlibrary{decorations.pathreplacing}

\usetikzlibrary{chains,fit,shapes}

\newtheorem{theorem}{Theorem}[section]
\newtheorem{lemma}[theorem]{Lemma}
\newtheorem{corolary}[theorem]{Corolary}

\newtheorem{property}[theorem]{Property}
\newtheorem{definition}{Definition}[section]
\newtheorem*{remark}{Remark}

\newcommand{\hide}[1]{}

\newcommand{\bra}[1]{{\langle{#1}\vert}}
\newcommand{\ket}[1]{{\vert{#1}\rangle}}

\newcommand{\bracket}[2]{\langle #1 \vert #2 \rangle}
\newcommand{\ketbra}[1]{\vert #1 \rangle \ \!\!\! \langle #1 \vert}

\newcommand{\cnot}			{{\textsf{c-Not}}}
\newcommand{\swap}		{{\textsf{Sw}}}
\newcommand{\qid}			{\mathds{1}}
\newcommand{\had}			{{\textsf{H}}}
\newcommand{\phase}		{{\textsf{S}}}

\newcommand{\qs}	{{\textsf{q}_\textsf{s}}}
\newcommand{\qh}	{{\textsf{q}_\textsf{h}}}
\newcommand{\blk}			{{\Box}}
\newcommand{\qA}			{{\mathbb{A}}}

\newcommand{\qU}			{{\mathbb{U}}}
\newcommand{\qD}			{{\mathbb{D}}}
\newcommand{\xL}			{{\textsf{L}}}
\newcommand{\xN}			{{\textsf{N}}}
\newcommand{\xR}			{{\textsf{R}}}

\newcommand{\smn}                  {{$s$-$m$-$n$}}

\DeclareMathOperator{\K}{K}

\DeclareMathOperator{\I}{I}
\DeclareMathOperator{\tr}{Tr}

\newcommand{\nats}                  {{\mathbb N}}

\newcommand{\qurtains}             {\hfill{QED}}

\title{Quantum Kolmogorov complexity and quantum correlations in deterministic-control quantum Turing machines}
\author[1,2]{Mariano Lemus}
\thanks{Corresponding author}
\email{mlemus@tecnico.ulisboa.pt}
\author[1,2]{Ricardo Faleiro}
\email{ricardofaleiro@tecnico.ulisboa.pt}
\author[1,2]{Paulo Mateus}
\email{paulo.mateus@tecnico.ulisboa.pt}
\orcid{0000-0002-2393-8224}
\author[1,2]{Nikola Paunkovi\'c}
\email{npaunkov@mathtecnico.ulisboa.pt}
\orcid{0000-0002-9345-4321}
\author[2,3,4]{Andr\'e Souto}
\email{ansouto@ciencias.ulisboa.pt}
\orcid{0000-0001-8792-959X}
\affiliation[1]{Departamento de Matemática, Instituto Superior Técnico, Universidade de Lisboa, \\Av.Rovisco Pais, 1049-001 Lisboa, Portugal}
\affiliation[2]{Instituto de Telecomunicações, Av.Rovisco Pais, 1049-001 Lisboa, Portugal}
\affiliation[3]{Lasige, Faculdade de Ciências da Universidade de Lisboa, Campo Grande, 1749-016 Lisboa, Portugal}
\affiliation[4]{Departamento de Informática,Faculdade de Ciências da Universidade de Lisboa,\\Campo Grande, 1749-016 Lisboa, Portugal}

\date{}
\begin{document}

\maketitle
\begin{abstract}
This work presents a study of Kolmogorov complexity for general quantum states from the perspective of deterministic-control quantum Turing Machines (dcq-TM). We extend the dcq-TM model to incorporate mixed state inputs and outputs, and define dcq-computable states as those that can be approximated by a dcq-TM. Moreover, we introduce (conditional) Kolmogorov complexity of quantum states and use it to study three particular aspects of the algorithmic information contained in a quantum state: a comparison of the information in a quantum state with that of its classical representation as an array of real numbers, an exploration of the limits of quantum state copying in the context of algorithmic complexity, and study of the complexity of correlations in quantum systems, resulting in a correlation-aware definition for algorithmic mutual information that satisfies symmetry of information property. 
\end{abstract}


\section{Introduction}
\label{sec:introduction}

Kolmogorov complexity is a mathematical formulation to capture the intuitive idea of the amount of information an individual object (usually a string) has. This notion was independently proposed by Kolmogorov in 1965~\cite{kol:65}, Solomonoff in 1964~\cite{sol:64} and Chaitin in 1966~\cite{cha:66}. It defines the amount of information of an object as the number of minimal instructions to be given to a machine to reproduce that object. If the object is a string $x$, and the machine is a Universal Turing machine $\mathsf{U}$, the Kolmogorov complexity of $x$ is the length of the shortest program that running in $\mathsf{U}$ as input, outputs $x$. We refer to the book by Li and Vit\'anyi \cite{li:vi:19} for a comprehensive reading of this theme. 

Several aspects must be considered to extend the notion of complexity to the quantum realm, as, in this context, the objects are quantum states, and the set of potential operations is not discrete.  In~\cite{vit:00}, Vit\'anyi identifies three approaches for the quantitative definition of information in a quantum state. The first one is the description (using classical bits) of the physical apparatus that outputs the given state. In this context, the continuous set of quantum states can be approximated by Cauchy sequences of states that can be directly outputted by the machine. The second approach consist on the \emph{qubit} description of the quantum machine outputting the state, which upper bounds the information in a given quantum state by the logarithm of the dimension of the associated Hilbert space. The third approach considers the information content in a quantum state to be the classical information content of the list of real numbers involved in a (fixed) mathematical expression for the state. Furthermore, there are additional aspects to consider, including how to handle the distinction between a machine outputting a state or a high fidelity approximation of it; and the issue of directly addressing mixed states, and therefore define the machine through density operators instead of pure states.

In the last few decades, several definitions and approaches have been proposed in the literature that explore all the aforementioned possibilities. Svozil~\cite{svo:96}, in his pioneer work, defined the quantum Kolmogorov complexity of a state as the length of the minimum classical description in a quantum Turing machine (in the machine model defined in~\cite{deu:85}) of that state. In~\cite{vit:01}, Vit\'anyi proposed the definition of the (quantum) Kolmogorov complexity of a pure quantum state $\ket \varphi$ as the minimum value of the sum of two parts: 
(i)~the length of a classical program computing $\ket\psi$, an approximation of $\ket\varphi$, 
and 
(ii)~the negative log-fidelity of the approximation of $\ket\psi$ to $\ket\varphi$.

These two first approaches considered classical descriptions for quantum states. In \cite{ber:dam:lap:01}, the authors considered the possibility of using a quantum program (encoded in qubits instead of bits) to define the complexity similarly to Vit\'anyi's.  In their work, Berthiaume, van Dam, and Laplante defined the Kolmogorov complexity of a qubit string as the length of the shortest quantum bit string that, given to a quantum universal Turing machine, produces the qubit string with high fidelity.  Its properties were further studied in \cite{mul:08, mul:07}.

Based on a generalization of the notion of universal semi-measure, G\'acs~\cite{gac:01} proposed to define Kolmogorov complexity based on density matrices. Later, in \cite{mor:bri:05}, the Kolmogorov complexity of a pure quantum state $\ket\varphi$ was defined as the length of the shortest description of a quantum circuit capable of producing such state (or an approximation of it). More recently, a new approach to define Kolmogorov complexity based on a particular type of quantum Turing machines, called deterministic controlled Turing machines was proposed~\cite{pmat:acs:asouto:17}. In this approach, a description of a quantum state $\ket{\varphi}$ is given by a sequence of quantum gates, which transforms an initial reference state of the quantum tape into the state $\ket{\varphi}$.

Although (classical) Kolmogorov complexity has been successfully applied to several fields ranging from security, computational complexity, bio-medicine,  etc~\cite{ant:etal:13,sou:etal:18,aze:etal:23}, corresponding applications of its quantum versions are yet just a few. For example, in \cite{mor:bri:05} the authors establish the ground to characterize quantum entanglement via Kolmogorov complexity. In~\cite{mor:bri:kra:06} Mora and Briegel and Kraus studied the application of quantum Kolmogorov complexity to thermodynamics and complexity theory. In~\cite{ben:kru:etal:06} and~\cite{miy:11} the authors explore the connection between (quantum) Kolmogorov complexity and entropy applied to Brudno's theorem and information-disturbance theorem, respectively. Another example of application is in cryptography, in~\cite{miy:ima:09}, the authors use (classical) Kolmogorov complexity to derive the security of BB84 quantum key distribution~\cite{ben:bra:84}. Recently, in~\cite{sar:al:ber:21}, algorithmic information theory has been used in genomics applications.

The interconnection between Kolmogorov Complexity Theory and Information Theory is well-known and extensively studied in the literature~\cite{gru:vit:08,ant:etal:11}. Concepts in one theory can find analogous versions in the other. One of the most important concepts in Information Theory is the notion of mutual information, which captures the inherent dependency of how much information is needed to explain one object having another object as a starting point. One of the most significant characteristics of (classical) mutual information is that it is symmetric, that is, $I(X,Y)=I(Y,X)$. On the other hand, for algorithmic mutual information such equality holds up to a logarithmic term of the size of the objects~\cite{zvo:lev:70,lee:rom:05,ant:etal:11}. The concept is briefly mentioned in the quantum Kolmogorov case in~\cite{vit:01}. In the context of quantum Kolmogorov complexity theory, the study of quantum algorithmic mutual information can help understand and quantify the complexity of the correlations between parts of multipartite quantum systems, since these systems can be inherently correlated. It is well-known that the correlations in quantum systems can be stronger than those present in classical systems due to the presence of entanglement~\cite{hor:09} or the more general quantum discord~\cite{mod:12}. Understanding quantum correlations and their relation to algorithmic mutual information measured in terms of a quantum version of Kolmogorov complexity is important, as it opens the possibility of quantifying the amount of manipulation required to create such correlations between two systems. Consequently, this connection can be useful in analyzing complexity in quantum computing, leading to more efficient algorithms, or in relating security and complexity in quantum cryptography protocols, for example.

\subsection{Overview of results}
The main results of this work can be summarized as follows:
\begin{itemize}
    \item We present a natural generalization of the \emph{deterministic-control quantum Turing machine} (dcq-TM) model, originally introduced in~\cite{pmat:acs:asouto:17}, to allow for mixed state inputs and outputs. 
    \item We define the set of dcq-computable states as the set of states which can be approximated with arbitrary precision using a dcq-TM.
    \item We define the (conditional) Kolmogorov complexity $\K$ of quantum states in the dcq-TM model and show that it is machine independent up to a constant.
    \item Using the concept of \emph{classical representation} of a quantum state to refer the array of real numbers describing its density matrix in the computational basis, we show that the defined set of dcq-computable states coincides with the set of states with computable classical representations. Moreover, we note that the algorithmic complexity of a quantum state and that of its classical representation are equal up to a constant.
    \item To contrast with the above result, we compare the properties of quantum states and their classical representations when used as a resource. In this context, we show that the representation contains more descriptive information when used as a resource for computation. We then further explore this difference by reducing the problem of computing an $n$-qubit state $\rho$ to the problem of computing a \emph{duplicate} of $\rho$ given access to one copy of it, resulting in a relation between the information in a single copy of a quantum state versus the information in two copies of itself. We conclude that for the majority of quantum states, cloning them is essentially as hard as to create them. The result can be interpreted as an algorithmic information version of the no-cloning theorem.
    \item Finally, we explore the concept of algorithmic mutual information for dcq-computable states. We show that the analogous version of the classical chain rule is not satisfied for general bipartite states and interpret this discrepancy to be caused by the presence of correlations among the two subsystems. We propose an alternative generalization of the chain rule for quantum systems and show that it is indeed satisfied under our definition of $\K$. Using this result, we present a correlation-accounting definition for algorithmic mutual information and show that it satisfies the symmetry of information property.
\end{itemize}

This paper is organized as follows: In Section~\ref{sec:preliminaries}, we briefly go through some important definitions on classical Kolmogorov complexity and computability of real numbers. In Section~\ref{sec:dcqtm}, we describe the dcq-TM machine model and present definitions for classical representations of quantum states, as well as classical simulators for quantum machines. In Section~\ref{sec:prepvscomp}, we define the concept of computable quantum states and compare it to the set of states with computable classical representations. In Section~\ref{sec:quantumk}, we define the concept of Kolmogorov complexity in the dcq-TM model and study some comparisons with the complexity of the associated classical representations, culminating in a dcq-TM version of no-cloning theorem. Later, in Section~\ref{sec:kcorrelations} we study potential definitions of algorithmic mutual information with the intent of quantifying quantum correlations, resulting in a generalization of the chain rule and an alternative definition for mutual information. Finally, in Section~\ref{sec:conclusionsk}, we summarize and discuss the results, as well as point to further research directions on the topic.

\section{Preliminaries}
\label{sec:preliminaries}

\subsection{Kolmogorov complexity}
\label{subsec:classicalk}

In this subsection we briefly review some basic properties of the standard classical algorithmic complexity $\K$ in order to better motivate the definitions in the following sections. Classically, the Kolmogorov complexity of a string $x$ is defined as the length of the shortest program that produces $x$ when given to a classical universal Turing machine.\footnote{In this work we focus on inputs to the Turing Machines which need not be prefix-free. In the literature, this is sometimes called ``plain'' Kolmogorov complexity.} We refer the reader to the book of Li and Vit\'anyi \cite{li:vi:19} for a complete study on this topic. 

\begin{definition} (Kolmogorov complexity) \\
\label{def:kcompclassical}
Let $x,y\in\{0,1\}^*$ be two strings and $\mathsf{U}$ a classical universal Turing machine. The Kolmogorov complexity of $x$ given $y$ is defined as:
\begin{equation}
\K(x\:|\:y) = \inf_{p}\{|p|: \mathsf{U}(p, y) = x\}.
\end{equation}
\end{definition}
The default value for $y$, the auxiliary input for the program $p$, is the empty string. In order to avoid overloaded notation, in those cases we typically drop this argument in the notation. Notice that Kolmogorov complexity is machine independent in the sense that we can fix a universal Turing machine $\mathsf{U}$ whose program size is at most a constant additive term worst than any other machine, and the running time is, at most, a logarithmic multiplicative factor slower than in any other machine (see Theorem 7.1.1 from~\cite{li:vi:19}).

In information theory, one of the fundamental quantities is mutual information. There is an analogous version for Kolmogorov complexity that we will call {\em algorithmic mutual information}, and is given by 
\begin{equation}
\label{eq:miclassicalk}
    \I_{\K}^{(1)}(x,y) = \K(x) - \K(x \:|\: y).
\end{equation}
In contrast with standard (Shannon) mutual information, this quantity is independent of any probability distribution and is concerned only with how much of the information in the string $y$ can be used to describe $x$. In information theory, one useful result is the chain rule, which allows us to express the entropy of a joint distribution in terms of its conditional and marginal distributions. For Kolmogorov complexity we have,
\begin{property}(Chain rule 1) \\
For all strings $x$ and $y$ in $\{0,1\}^*$,
\begin{equation}
\label{eq:chainrulek}
\K(x,y)=\K(x)+\K(y\:|\:x)+O(\log(\K(x,y))).
\end{equation}
\end{property}

Using Equation~(\ref{eq:chainrulek}) in combination with Equation~(\ref{eq:miclassicalk}) leads to a result called \emph{symmetry of information}, which states that
\begin{equation}
     \label{eq:siminf}
    \I_{\K}^{(1)}(x,y) = \I_{\K}^{(1)}(y,x) + O(\log(\K(x,y))).
\end{equation}
The chain rule also allows us to relate $I_{\K}^{(1)}$ with another quantity, which can also be understood as a measure of mutual descriptive information
\begin{equation}
    \I_{\K}^{(2)}(x,y) = \K(x) + \K(y)- \K(x, y).
\end{equation}
Analogous expressions for $\I_{\K}^{(1)}$ and $\I_{\K}^{(2)}$ exist in classical information theory for random variables $X,Y$ by replacing the Kolmogorov complexity $\K$ by the Shannon entropy function $H$. In classical information theory those quantities are proven to be equivalent~\cite{sha:48}. In the case of algorithmic information theory, they are related by 
\begin{equation}
    \label{eq:classicali1i2}
     \I_{\K}^{(1)}(x,y) = \I_{\K}^{(2)}(x,y) + O(\log(\K(x,y))).
\end{equation}
Let $n = |x| + |y|$, by noting that $\K(x,y) \leq \K(xy) + O(\log(n))$ we can slightly modify the chain rule to a form that will be later more suitable for generalizing to the quantum case~\cite{li:vi:19}: 

\begin{property}(Chain rule 2) \\
For all strings $x$ and $y$ in $\{0,1\}^*$ with $n = |x| + |y|$,
\begin{align}
    \label{eq:chainrulek2}
    \K(xy) &= \K(x)+\K(y\:|\:x)+O(\log(n \K(xy))) \\
    \label{eq:chainrulek21}
    &= \K(x)+\K(y\:|\:x)+O(\log(n)).
\end{align}
\end{property}
Even though Eq.~\eqref{eq:chainrulek21} is the simplest form of the classical chain rule, we want to highlight Eq.~\eqref{eq:chainrulek2} because it will allow us to better compare it with its quantum counterpart later on.

\subsection{Computability of real numbers}
\label{sec:realcomputability}
In this section, we briefly review the notions associated with the computability and algorithmic complexity of real numbers. Throughout the paper, we shall need to encode tuples of strings into a single string. Therefore, we  denote by $[x_1, \ldots, x_m]$  the encoding of the $m$-tuple $(x_1, \ldots, x_m)$ of binary strings into a single binary string. Fix a one-to-one encoding of natural numbers in strings and encoding of rational numbers by assigning the information of sign, numerator, and denominator to the different components of a tuple.

\begin{definition} 
\label{def:realcomputability}
A real number $r$ is said to be computable if there exists a (classical) Turing machine $\mathsf{T}$ which on any input $k \in \mathds{N}$ outputs $q_k \in \mathds{Q}$, such that
\begin{equation}
    \label{eq:realcomputability}
    | r - q_k | \leq \frac{1}{k}.
\end{equation}
In this case, we say that $\mathsf{T}$ computes $r$.
\end{definition}

Computable real numbers can be encoded effectively in natural numbers. Given an enumeration of the set of Turing machines, we will consider that a computable real number is encoded by any index that represents a Turing machine that will output the approximation of it according to Def.~\ref{def:realcomputability}. A tuple of real numbers $(r_1, \ldots, r_n) \in \mathds{R}^{n}$ is computable if all its components $r_i$ are computable. This encompasses complex numbers (ordered pairs of real numbers) and, notably, density matrices (arrays of complex numbers). We can proceed now to define the Kolmogorov complexity of a real number by identifying it with the length of a program that produces a sequence that converges to it according to Equation~\eqref{eq:realcomputability}: 
\begin{definition}
\label{def:kreals}
Let $r$ be a real number and $\mathsf{U}$ a universal classical Turing machine. The Kolmogorov complexity of $r$ is given by
\begin{equation}
    \K(r) = \inf_{p}\left\{|p|: \mathsf{U}_p \;\textnormal{computes}\; r )\right\},
\end{equation}
where $\mathsf{U}_p(x) = \mathsf{U}(p,x)$.
\end{definition}
\noindent The notion Kolmogorov complexity generalizes straightforwardly to arrays of real numbers and matrices as expected.

\begin{remark}
Since the set of computable numbers contains $\mathbb{N}$, some ambiguity may arise when talking about the complexity of a natural number, as both Def.~\ref{def:kcompclassical} and Def.~\ref{def:kreals} may apply (depending on the chosen encoding), but they in general do not coincide. In this document, unless stated otherwise, whenever a number $x$ is understood to take values only in $\mathbb{N}$ (index in an enumeration, counter in a \emph{for} loop, etc.) we assume that $\K(x)$ is given by  Def.~\ref{def:kcompclassical}. Otherwise, we assume its complexity is given by  Def.~\ref{def:kreals}.
\end{remark}

\section{Definition of the machine}
\label{sec:dcqtm}

\subsection{Deterministic control quantum Turing machine (dcq-TM)}
A dcq-TM is a completely deterministic machine in the sense that the flow of computation is entirely classically controlled, but it is also capable of having quantum states as inputs and outputs. The following description of a dcq-TM is based on the one defined in~\cite{pmat:acs:asouto:17} with the added feature that the machine is also able to handle general mixed states, and not just pure states.  The machine is understood as having at least two infinite tapes, one classical and one quantum, with their respective classical and a quantum heads\footnote{Please note that our ``quantum head'' is in fact a classical head that operates over the quantum tape. In this sense, it evolves according to a completely classical-deterministic transition function, as opposed to the Deutsch and Bernstein-Vazirani models \cite{deu:85, ber:vaz:97} where the quantum head evolves unitarilly and can be found in superpositions of its classical internal states.}. The distinctive feature of the machine is that it has a deterministic control over the computation happening on the quantum tape -- this control is only dependent on the set of internal states ($Q$) of the dcq-TM, which are classical, and the contents of the classical tape. Thus, the set of internal states  and the classical input specify the dcq-TM's computational dynamics entirely, and the quantum tape is a ``work only" tape whose contents do not influence the computation. The quantum tape is effectively an infinite quantum memory, composed by quantum cells, each associated with a two dimensional (qubit) Hilbert space $\mathds{C}^{2}$. The classical control of the dcq-TM commands the quantum head in a deterministic way. It chooses where the quantum head moves to and picks, from a pre-chosen universal set $\qU$ (which without loss of generality is assumed to be finite), which unitary gates to apply and to which cells of the quantum tape to apply them to. It is important to stress that the unitary transformations in $\qU$ must be described only with computable complex numbers, otherwise the machine would compute super-Turing distribution/functions. The machine is rigorously defined as follows,
\begin{definition}
\label{def:dcqtm}
  A deterministic-control Quantum Turing machine or dcq-TM is defined by the pair, $(Q,\delta)$, where $Q$ and $\delta$ are, respectively, the finite set of control states containing at least the two distinct states $\qs$ (starting state) and $\qh$ (halting state) and the transition function of the computation. 
  The transition function is, 
  \begin{equation}
\delta: Q \times \qA \rightarrow \qU \times \qD \times \qA \times \qD \times Q,
 \end{equation}
 where $\qA = \{0, 1 ,\blk\}$ is the alphabet of the classical tape, $ \qU$ is a universal set of unitary operators that can be applied to the quantum tape, and $\qD=\{\xL,\xN,\xR\}$
is set of possible head displacements (left, none, right). 
\end{definition}

\subsubsection{Anchor cell and input/output schemes}
In order to have the quantum and classical inputs/outputs properly defined in the dcq-TM we introduce the concept of an \textit{anchor cell}. The anchor cell serves the purpose of specifying a global reference point for each tape, in order to identify unambiguously the inputs and outputs of the machine. For a given tape, we specify the anchor cell (on that tape) to be the first cell before the input starts i.e if we index the input cells with positive numbers $i>0$, the anchor cell  will always be the cell with index $i=0$.
The input and the output schemes of the dcq-TM presented here are slightly different from the ones in the definition of \cite{pmat:acs:asouto:17}, where there is no \textit{anchor cell} defined, and the machine is not aware of the size of its quantum input/output. In the spirit of generalizing the model for general mixed states and more naturally defining inputs and outputs that may start or end with $\ket{0}$ states. We introduce this alternative definition from which it is straightforward to recover all the results already established regarding the dcq-TM. 
 
\textbf{Inputs} -- We say the machine received as input the pair $(x;\rho_{\langle1,n\rangle})$ if its internal state is $\qs$, its classical and quantum heads are parked in their respective anchor cells, and the tapes' contents are as shown in Figure~\ref{fig:inputs}. The classical input $x$ is the string written between the anchor and the first blank after the anchor. At the end of the classical input, we insert a blank cell to distinguish strings $x$ and $n$,  where the latter specifies the quantum input in the quantum tape. The quantum input $\rho_{\langle1,n\rangle}$ is the physical quantum state present in the set of cells between the anchor and the cell specified by string $n$. 
 
\begin{figure}[H]
\begin{center}
\begin{tikzpicture}
\tikzstyle{every path}=[thick]

\edef\sizetape{0.7cm}
\tikzstyle{tmtape}=[draw,minimum size=\sizetape]
\tikzstyle{tmhead}=[arrow box,draw,minimum size=.33cm,arrow box
arrows={north:.23cm}]

\begin{scope}[xshift=.6cm,start chain=1 going right,node distance=-0.15mm]
   \node [on chain=1,tmtape,draw=none] {$\ldots$};
    \node [on chain=1,tmtape] (input) {$\Box$};
    \node [on chain=1,tmtape] (input) {$\Box$};
    \node [on chain=1,tmtape] {$x_1$};
    \node [on chain=1,tmtape] {...};
    \node [on chain=1,tmtape] {$x_k$};
    \node [on chain=1,tmtape] {$\Box$};
    \node [on chain=1,tmtape] {$n_1$};
    \node [on chain=1,tmtape] {...};
    \node [on chain=1,tmtape] {$n_s$};
    \node [on chain=1,tmtape] {$\Box$};
    \node [on chain=1,tmtape] {$\Box$};
    \node [on chain=1,tmtape,draw=none] {$\ldots$};
    \node [on chain=1] {\textbf{Classical Tape}};
\end{scope}

\node [tmhead,yshift=-.2cm, color=black] at (input.south) (head) {};

\node [yshift=1cm] at (input.south)  {\footnotesize $\textup{Anchor}_c$};


\begin{scope}
[xshift=0.7cm,yshift=-2.5cm, start chain=1 going right  ,node distance=-0.15mm]
    \node [on chain=1,tmtape,draw=none] {$\ldots$};
    \node [on chain=1,tmtape] {$\ket{0}\!\bra{0}$};
    \node [on chain=1,tmtape] (input) {$\ket{0}\!\bra{0}$};
    \node [on chain=1,tmtape] {$\rho_1$};
    \node [on chain=1,tmtape] {...};
    \node [on chain=1,tmtape] {...};
    \node [on chain=1,tmtape] {$\rho_n$};
    \node [on chain=1,tmtape] {$\ket{0}\!\bra{0}$};
    \node [on chain=1,tmtape] {$\ket{0}\!\bra{0}$};
    \node [on chain=1,tmtape,draw=none] {$\ldots$};
    \node [on chain=1] {\textbf{Quantum Tape}};
\end{scope}

\node [tmhead,yshift=-.2cm, color=gray] at (input.south) (head) {};

\node [yshift=1cm] at (input.south)  {\footnotesize $\textup{Anchor}_q$};

			 \draw [decorate,decoration={brace, mirror,amplitude=5pt},xshift=10pt,yshift=55pt]
(2.05,-2.5) -- (4.20,-2.5)  node [black,midway,yshift=-20pt] {\!\!$x$: \footnotesize{\!\!C-Input}}; 
 \draw [decorate,decoration={brace, mirror,amplitude=5pt},xshift=22pt,yshift=55pt]
(4.49,-2.5) -- (6.5,-2.5)  node [black,midway,yshift=-20pt] {\ \ \ \ \ $n$: \footnotesize{\!\!Q-Input size specification}}; 

\draw [decorate,decoration={brace, mirror,amplitude=5pt},xshift=25pt,yshift=0pt]
(2.25,-3) -- (5.05,-3) node [black,midway,yshift=-20pt] {$\rho_{\langle1,n\rangle}$: \footnotesize{\!\!Q-Input}}; 

\end{tikzpicture}
\end{center}{}
\vspace{-0.4cm}
\caption{\textbf{Starting configuration of the dcq-Turing machine with a classical and quantum input.} $x = x_{1}\cdot x_{2} \cdots x_{k}$ is the classical input, while $\rho_{\langle1,n\rangle}$ is the quantum input ``written" on $n$ cells of the quantum tape, where $\rho_{i}$ is the local state of the $i-$th cell, represented by $ \mathbold{\rho}_i=\textnormal{Tr}_{(1,\ldots, i-1, i+1, \ldots, n)}[\mathbold{\rho}_{\langle1,n\rangle}]$. The string $x\Box n$ is surrounded by blank symbols extending to infinity in both directions.  The quantum input is surrounded by cells in the local zero state $\ket{\text{0}}\!\bra{\text{0}}$ extending to infinity in both directions.}
\label{fig:inputs}
\end{figure}

\begin{figure}[H]
\begin{center}
\begin{tikzpicture}
\tikzstyle{every path}=[ thick]

\edef\sizetape{0.7cm}
\tikzstyle{tmtape}=[draw,minimum size=\sizetape]
\tikzstyle{tmhead}=[arrow box,draw,minimum size=.33cm,arrow box
arrows={north:.23cm}]

\begin{scope}[start chain=1 going right,node distance=-0.15mm]
    \node [on chain=1,tmtape,draw=none] {$\ldots$};
    \node [on chain=1,tmtape] (input) {$\Box$};
    \node [on chain=1,tmtape] (input) {$\Box$};
    \node [on chain=1,tmtape] (input) {...};
    \node [on chain=1,tmtape] (input) {$w$};
    \node [on chain=1,tmtape] {$y_1$};
    \node [on chain=1,tmtape] {...};
    \node [on chain=1,tmtape] {$y_l$};
    \node [on chain=1,tmtape] {$\Box$};
    \node [on chain=1,tmtape] {$m_1$};
    \node [on chain=1,tmtape] {...};
    \node [on chain=1,tmtape] {$m_t$};
    \node [on chain=1,tmtape] {$\Box$};
    \node [on chain=1,tmtape] {$w$};
    \node [on chain=1,tmtape] {...};
    \node [on chain=1,tmtape] {$\Box$};
    \node [on chain=1,tmtape] {$\Box$};
    \node [on chain=1,tmtape,draw=none] {$\ldots$};
    \node [on chain=1] {\!\!\textbf{Classical Tape}};
\end{scope}

\node [tmhead,yshift=-.2cm, color=black] at (input.south) (head) {};

\node [yshift=1cm] at (input.south)  {\footnotesize $\textup{Anchor}_c$};


\begin{scope}
[xshift=1cm, yshift=-2.5cm, start chain=1 going right  ,node distance=-0.15mm]
    \node [on chain=1,tmtape,draw=none] {$\ldots$};
    \node [on chain=1,tmtape] {$\ket0\!\bra0$};
    \node [on chain=1,tmtape] {$\ket0\!\bra0$};
    \node [on chain=1,tmtape] {$...$};
    \node [on chain=1,tmtape] (input) {$\mu$};
    \node [on chain=1,tmtape] {$\sigma_1$};
    \node [on chain=1,tmtape] {...};
    \node [on chain=1,tmtape] {$\sigma_m$};
    \node [on chain=1,tmtape] {$\mu$};
    \node [on chain=1,tmtape] {$...$};
    \node [on chain=1,tmtape] {$\ket0\!\bra0$};
    \node [on chain=1,tmtape] {$\ket0\!\bra0$};
    \node [on chain=1,tmtape,draw=none] {$\ldots$};
    \node [on chain=1] {\!\!\textbf{Quantum Tape}};
\end{scope}

\node [tmhead,yshift=-.2cm, color=gray] at (input.south) (head) {};

\node [yshift=1cm] at (input.south)  {\footnotesize $\textup{Anchor}_q$};

			 \draw [decorate,decoration={brace, mirror,amplitude=5pt},xshift=30pt,yshift=55pt]
(2.2,-2.5) -- (4.3,-2.5)  node [black,midway,yshift=-20pt] {\!\!$y$: \footnotesize{\!\!C-Output}}; 

\draw [decorate,decoration={brace, mirror,amplitude=5pt},xshift=119pt,yshift=55pt]
(1.95,-2.5) -- (4.1,-2.5)  node [black,midway,yshift=-20pt] {\ \ \ \ \ $m$: \footnotesize{\!\!Q-Output size specification}}; 

\draw [decorate,decoration={brace, mirror,amplitude=5pt},xshift=22pt,yshift=-1pt]
(4.1,-3) -- (6.2,-3) node [black,midway,yshift=-20pt] {$\sigma_{\langle 1,m\rangle}$: \footnotesize{\!\!Q-Output}}; 

\end{tikzpicture}
\end{center}{}
\vspace{-0.4cm}
\caption{\textbf{Final state of the dcq-Turing machine with a classical and quantum outputs.} $y = y_{1}\cdot y_{2} \cdots y_{r}$ is the classical output, while $\sigma_{\langle1,s\rangle}$ is the quantum output, where $\sigma_{i}$ is the local state of the $i-$th cell, represented by $ \mathbold{\sigma}_i=\textnormal{Tr}_{(1,\ldots, i-1, i+1, \ldots, n)}[\mathbold{\sigma}_{\langle1,s\rangle}]$. The size (in number of qubits) of $\sigma_{(1,s)}$ is specified by the string $s$ in the classical tape. The working cells, which are used during the computation but are not part of the output in the classical and quantum tape are denoted by $w$ and $\mu$, respectively.}
\label{fig:outputs}
\end{figure}
\textbf{Outputs} -- We say the machine produced as output, $(y;\sigma_{\langle 1,s\rangle})$, if the machine is in the halting state $\qh$ and the classical and quantum tapes are in the configuration shown in Figure~\ref{fig:outputs}. The classical output $y$ is defined as the string written in the set of cells that are between the anchor and first blank to the right of the anchor. The quantum output specification $m$ is defined as the string between the first and the second blanks appearing after the anchor. Finally, the quantum output is defined as the physical quantum state present in the set of cells between the anchor and the cell specified by string $m$. 

Although there are many ways to evaluate functions that have multiple classical inputs/outputs using the above configuration, for the purposes of the examples of this paper it is convenient to have a unified encoding to use whenever we refer to algorithms that use a single string $x$ to encode a tuple $(x_1,\ldots, x_k)$ of strings. We use the following encoding
\begin{equation}
    (x_1,\ldots, x_k) \equiv x_1 \cdot \ldots \cdot x_k \equiv \ell(x_1) 0 x_1 \ell(x_2)0 x_2\cdots \ell(x_{k-1})0 x_{k-1} x_k,
\end{equation}
where $\ell(x_i) = 1^{|x_i|}$ is a string of 1s of length $|x_i|$. Using this encoding, each input $x_i$ adds $2|x_i|+1$ bits to the input string, except the last one, which only adds $|x_i|$; that is
\begin{equation}
    \left|(x_1,\ldots, x_k)\right| = \sum_{i=1}^{k-1}2|x_i|+|x_k|+(k-1).
\end{equation}
Multiple quantum inputs/outputs can be defined by encoding the tuple of sizes (number of qubits) of the different systems in the same way described above.

\subsubsection{Symbolic notation for dcq-TM computations}
Thus, an input on the dcq-TM of the form $(x;\rho_{\langle1,n\rangle})$, specifies the following representation of the tapes,
\begin{equation}
\begin{split}
    &[\textrm{Classical tape}] = [\Box]_{(-\infty,0)}\;[x]\;[\Box]\;[n]\;[\Box]_{(|x|+|n|+1,\infty)} \\ \\
   & \rho_{\;\textrm{Quantum Tape}} = \ket{\textbf{0}}\!\bra{\textbf{0}}_{(-\infty,0)}\otimes \rho_{\langle1,n\rangle} \otimes \ket{\textbf{0}}\!\bra{\textbf{0}}_{(n+1,\infty)},
    \end{split}{}
\end{equation}
where $\ket{\textbf{0}}\!\bra{\textbf{0}}_{(a,b)}$ denotes $\ket{0}\!\bra{0}_a\otimes \ket{0}\!\bra{0}_{a+1} \otimes \ldots \otimes \ket{0}\!\bra{0}_{b}$. If a specific dcq-TM machine, $\mathsf{T}$, on input $(x;\rho_{\textnormal{in}})$ produces output $(y;\rho_{\textnormal{out}})$ we represent the computation symbolically as, 
\begin{equation}
    \mathsf{T}(x;\rho_{\textnormal{in}})=(y;\rho_{\textnormal{out}}),
\end{equation}
for the overall output, and
\begin{equation}
   \mathsf{T}^{(C)}(x;\rho_{\textnormal{in}}) = y \quad \quad  \mathsf{T}^{(Q)}(x;\rho_{\textnormal{in}}) = \rho_{\textnormal{out}}, 
\end{equation}
for the classical and quantum outputs. To represent empty inputs/outputs we use the symbols $\blk$ for the classical tape, and $\epsilon$ for the quantum one. Note that $\epsilon$ is used to denote quantum input/outputs of size zero. We call the reader's attention to the fact that $\epsilon$ is not a density operator acting on a Hilbert space, since a zero size input/output does not refer to any system in the quantum tape; we use it to refer to the cases where there may not be a quantum input/output.

The computability properties of dcq-TM may be dependent on the chosen set $\qU$ of basic quantum operations, namely the states that can be prepared exactly by the machine. On this work, we will focus on dcq-TM based on the universal set $\qU = \{\qid,\had,\phase,\frac{\pi}{8},\swap,\cnot\}$ commonly used in the circuit model for quantum computation. This set is important because it is associated with the capabilities of practical quantum computers. 
It is also important to remark that the transition function of a dcq-TM depends only on the contents of the classical tape, this is the reason the model is denoted as {\em deterministic control}. In \cite{pmat:acs:asouto:17}, the universality of this model was proven and, moreover, it was also shown that the dcq-TM model could simulate a special type of quantum circuit called a \textit{seesaw circuit}. A \textit{polynomial translatability theorem} was proven which states that any general quantum circuit can be simulated by one such seesaw circuit with a polynomially bounded overhead on the number of gates,  showing that the dcq-TM can solve in a bounded polynomial-time the same problems that the regular quantum circuit model can, therefore, the BQP complexity class coincides with the analogous dcBQP class defined for the dcq-TM model. We also know that there is a universal dcq-TM, $\mathsf{T_U}$, that can simulate every other dcq-TM on any arbitrary input. As such,  without loss of generality, 
\begin{equation}
    \mathsf{T_U}(x ;\rho_{\textnormal{in}})=(\blk ;\rho_{\textnormal{out}}) \Leftrightarrow \mathcal{C}_x^{*}(\rho_{\textnormal{in}})= \rho_{\textnormal{out}}.
\end{equation}
The previous expression represents symbolically the equivalence between the computation taken on a specific input in the dcq-TM model, and it is implementation of the \textit{seesaw circuit} $\mathcal{C}_x^{*}$ that gives the same output. From now on, whenever we refer to a universal dcq-TM, we mean a universal dcq-TM machine fulfilling the s-m-n property.

Sometimes, we may want to refer to an instance of a universal dcq-TM $\mathsf{U}$, simulating some other dcq-TM $\mathsf{T}$. From the s-m-n property, we know that there is a string $t \in \lbrace 0,1 \rbrace^*$ such that, for any input $(x;\rho_{\textnormal{in}})$ we have that $\mathsf{U}\bigl( (t,x);\rho_{\textnormal{in}} \bigr) = \mathsf{T}(x;\rho_{\textnormal{in}})$. In these cases we say that the machine $\mathsf{U}$ \emph{runs the program} $t$ \emph{on input} $(x;\rho_{in})$, denoted symbolically by:
\begin{equation}
    \mathsf{U}_{t}(x;\rho_{\textnormal{in}}) =
    \mathsf{U} \bigl( (t,x);\rho_{\textnormal{in}} \bigr) = \mathsf{T}(x;\rho_{\textnormal{in}}).
\end{equation}

\subsection{Classical quantum state encoding and simulators}

In this work we are interested in comparing the complexity aspects of two related, but fundamentally different types of objects. The first are the quantum states themselves, understood as density operators acting on a given Hilbert space, encoded in the physical states of the cells of the quantum tape, and are denoted by standard Greek letters $\rho, \sigma$, etc. The second are the \emph{classical representations} of the quantum states, which are understood as arrays of complex numbers describing the respective density matrices in the computational basis, encoded by strings of symbols in the classical alphabet $\qA$ according to an appropriate rule (as described in Section~\ref{sec:realcomputability}), and denoted by bolded Greek letters $\mathbold{\rho}, \mathbold{\sigma}$, etc.
In the context of inputs and outputs of dcq-TMs, a statement of the form
\begin{equation}
    \mathsf{T}(\mathbold{\rho}_{\textnormal{in}};\epsilon) = (\mathbold{\rho}_{\textnormal{out}};\epsilon),
\end{equation}
is interpreted as follows: ``Upon classical input of any string encoding the computational-basis matrix associated to the state $\rho_{\textnormal{in}}$ and no quantum input, the dcq-TM $\mathsf{T}$ outputs a string encoding the respective matrix for the state $\rho_{\textnormal{out}}$ without a quantum output". Finally, for simplicity, following the relation between the states and their matrix representations, for every function $f$ of quantum states $\rho, \sigma, \dots$, we define $f(\mathbold\rho, \mathbold\sigma, \dots) \equiv f(\rho, \sigma, \dots)$. 
 
Before moving on to discuss directly dcq-computable and (plain) dcq-computable quantum states, we establish a notation for the classical simulators of these machines. For any dcq-TM $\mathsf{T}$, consider a ``classical'' machine $\Tilde{\mathsf{T}}$ (a dcq-TM that does not interact with the quantum tape and whose quantum output is $\epsilon$ for any input), defined constructively from $\mathsf{T}$, which appropriately approximates the matrix operations associated with the logic gates applied during the execution of $\mathsf{T}$ such that 

\begin{equation}
    \label{eq:classicalsimulator}
    \mathsf{T}(x;\rho_{\textnormal{in}}) = (y; \rho_{\textnormal{out}}) \Leftrightarrow \Tilde{\mathsf{T}} \bigl( (x, \mathbold{\rho}_{\textnormal{in}}) ; \epsilon \bigr)= \bigl( (y, \mathbold{\rho}_{\textnormal{out}}) ; \epsilon \bigr),
\end{equation}
whenever the respective string encodings $\mathbold{\rho}_{\textnormal{in}}$ and $\mathbold{\rho}_{\textnormal{out}}$ exist (that is, whenever the matrices $\mathbold{\rho}_{\textnormal{in}}$ and $\mathbold{\rho}_{\textnormal{out}}$ are computable).


\section{Computability of quantum states}
\label{sec:prepvscomp}

\subsection{Directly computable states}

\begin{definition}
\label{def:prepstate}
A quantum state $\rho$ is  \textit{directly dcq-computable} whenever it can be outputted by a dcq-TM, provided the machine started the computation with no auxiliary input i.e.,
\begin{equation}
 \exists \mathsf{T}: \mathsf{T}(\blk;\epsilon)=(\blk;\rho).
\end{equation}
The set of \textit{directly dcq-computable} states is denoted by $\mathsf{DCOMP_{q}}$.
\end{definition}

Denote by $\mathsf{T}(x;\sigma)\!\!\downarrow _{t}$ a particular (defined constructively from $\mathsf{T}$ and $t$) dcq-TM that simulates $\mathsf{T}(x;\sigma)$, except that if at the step $t$ it has not halted, if forces it to halt, then outputs the output of the simulated machine (according to the rules specified in Section~\ref{sec:dcqtm}). $\mathsf{T}^{(C)}(x;\sigma)\!\!\downarrow_{t}$ and $\mathsf{T}^{(Q)}(x;\sigma)\!\!\downarrow_{t}$ are defined analogously.
Consider now Algorithm~\ref{alg:prepstategenerator} which, given a universal dcq-TM $\mathsf{U}$, uses Cantor's zig-zag method to assign to each non-negative integer $s$ a different directly dcq-computable state associated with running $\mathsf{U}$ for some combination of input and number of steps, which is different for every value of $s$. Because the output of a dcq-TM is well defined for any state of its tapes, each value of $s$ is uniquely associated with a directly dcq-computable state through Algorithm~\ref{alg:prepstategenerator}. Additionally, because every directly dcq-computable state $\rho$ is associated with $\mathsf{U}$ running \emph{some} program for \emph{some} number of steps, we know that there exists a value of $s$ for which Algorithm~\ref{alg:prepstategenerator} outputs $\rho$. This means that the set of states that can be returned by Algorithm~\ref{alg:prepstategenerator} coincides set of directly dcq-computable states. For a fixed reference universal dcq-TM, the quantum output of Algorithm~\ref{alg:prepstategenerator} defines a surjective function $\Pi_1(s)$ from the set of non-negative integers to $\mathsf{DCOMP_{q}}$, that is, an enumeration of it.

\begin{center}
\begin{algorithm}[H]
\caption{directly dcq-computable state generator algorithm.}
\label{alg:prepstategenerator}
 \textbf{Parameters:} Universal dcq-TM $\mathsf{U}$ \; 
 \textbf{Input:} Integer $s \geq 0$\; 
 Set $n$ equal to the smallest natural number such that $s < {(n+1)(n+2)}/{2}$\;
 \eIf{$n \leq 1$}{$m=s-n$\;}{$m = s \mod ({n(n+1)}/{2})$\;}
 $\sigma = \mathsf{U}^{(Q)}(n-m;\epsilon)\!\!\downarrow_{m}$\;
 \Return{$(\blk; \sigma)$}\;
\end{algorithm}
\end{center}

\subsection{Computable states}
\begin{definition} 
\label{def:compstate}
A quantum state $\rho$ is \textit{dcq-computable} if there exists a dcq-TM $\mathsf{T}$ and an infinite sequence $\lbrace \sigma_i \rbrace_i$ of directly dcq-computable states such that:
    \begin{enumerate}
        \item $\mathsf{T}(k;\epsilon) = (\blk; \sigma_k)$, i.e., upon classical input $k$ and no quantum input, the machine $T$ outputs the state $\sigma_k$.
        \item $D(\rho, \sigma_k) \leq {1}/{k}$, where $D(\cdot)$ denotes the trace distance of operators acting on the respective Hilbert space.
    \end{enumerate}
We denote the set of all dcq-computable quantum states by $\mathsf{COMP_{q}}$.
\label{def:qcomp}
\end{definition}

Consider the set $\mathsf{COMP_{c}}$ of quantum states for which their density matrix representation, when written in the computational basis, has computable complex numbers as components. This set is relevant because it represents the set of quantum states that can be written in a classical computer. Let us state a pair of algorithms that will be useful in finding relationships between quantum states and their representation in the context of the dcq-TM model. Since our machine model is controlled by classical programs, we would expect that $\mathsf{COMP_{q}}$, the set of computable states as defined by Def.~\ref{def:qcomp} and the set $\mathsf{COMP_{c}}$ coincide. This is indeed the case as shown below.

\begin{theorem}
\label{thm:CQequivalence}
$\mathsf{COMP_{q}} = \mathsf{COMP_{c}}$
\end{theorem}
\begin{proof}\ \\
($\mathsf{COMP_{q}} \subseteq \mathsf{COMP_{c}}$):

\noindent Let $\rho \in \mathsf{COMP_{q}}$, and  $\lbrace \sigma_i \rbrace_i$ a sequence which satisfies Def.~\ref{def:qcomp}. From the definition of $\mathsf{COMP_{c}}$, to show that $\rho \in \mathsf{COMP_{c}}$ we must show that $\mathbold{\rho}$ is computable. Note that the set $\mathds{U}$ consists on matrices with computable components. It follows that the associated sequence of matrices $\lbrace \mathbold{\sigma}_i \rbrace_i$ is computable since, for any dcq-TM $\mathsf{T}$ such that $\mathsf{T}(k;\epsilon) = (\blk; \sigma_k)$, the classical machine $\Tilde{\mathsf{T}}$ simulates the evolution of the initial state though the quantum circuit specified by $\mathsf{T}$ and computes the components $\mathbold{\sigma}_n$. By recalling that the limit of any classically computable sequence is computable, it follows that $\lim_{n \rightarrow \infty} \mathbold{\sigma}_n = \mathbold{\rho}$ is computable and thus, $\rho \in \mathsf{COMP_{c}}$. 
\vspace{5mm} 

\noindent($\mathsf{COMP_{c}} \subseteq \mathsf{COMP_{q}}$): \\
Let $\mathbold{\rho}$ be a density matrix with computable components. From the definition of computability of real numbers, this means that there exists a set $\left\{ f_{rs}^k \right\}_{r,s,k}$, of computable functions over $\mathds{N}$ such that the sequence of matrices ${\mathbold{\sigma}_k} = {(g_{rs}(k))}$, where $g_{rs}(k)= {f_{rs}^0(k)} + i {f_{rs}^1(k)}$, satisfies
\begin{equation}
    D(\mathbold{\sigma}_k, \mathbold{\rho}) \leq \frac{1}{k}.
\end{equation}

\begin{center}
\begin{algorithm}[H]
\caption{Given a density matrix $\mathbold{\sigma}$ and integer $k$, finds a directly dcq-computable state that is closer to $\sigma$ than ${1}/{(2k)}$.}
\label{alg:classicaltoquantum}
\textbf{Parameters:} Universal dcq-TM $\mathsf{U}$ \;
\textbf{Input:} Integer $k$, $2^{n} \times 2^{n}$ density matrix $\mathbold{\sigma}$\; 
$s \leftarrow 0$\;
\While{$\mathbold{\sigma}' = \Tilde{\Pi}_1(s)$ is not an $2^{n} \times 2^{n}$ matrix such that $D(\mathbold{\sigma}' , \mathbold{\sigma} ) \leq {1}/{(2k)}$}{
    $s \leftarrow s+1$\;
}
$\sigma' \leftarrow \Pi_1(s)$\;
\Return{$(s;\sigma')$}
\end{algorithm}
\end{center}

We want to use $\lbrace \mathbold{\sigma}_k \rbrace$ to construct a (computable) sequence $\lbrace \sigma'_k \rbrace$ of states associated with directly dcq-computable states that satisfy Def.~\ref{def:qcomp}. Consider the following construction for $\sigma'_k$: Fix a universal dcq-TM $\mathsf{U}$ and run Algorithm~\ref{alg:classicaltoquantum} with input $((k, \mathbold{\sigma}_{2k});\epsilon)$ and define $\sigma'_k$ as its quantum output. We know that Algorithm~\ref{alg:classicaltoquantum} will always find such state because the set of gates that a dcq-TM has access is universal, and hence the set of directly dcq-computable states is dense in the set of quantum states. To prove that $\rho \in \mathsf{COMP_{q}}$ it only remains to show that the sequence $\lbrace \sigma_i \rbrace$ satisfies Def.~\ref{def:qcomp}(2). Indeed, we have that
\begin{equation}
    D( \sigma_{k}' , \rho) \leq D( \sigma_{k}' ,  \sigma_{2k} ) + D( \sigma_{2k} , \rho) \leq \frac{1}{2k} + \frac{1}{2k} = \frac{1}{k}.
\end{equation}

\end{proof}

\section{Algorithmic complexity in the dcq-TM model}
\label{sec:quantumk}

Having defined the features of the dcq-TM, we proceed to introduce two (non-equivalent) notions of the concept of Kolmogorov complexity. The first one is a generalization of the one proposed in~\cite{pmat:acs:asouto:17} for general directly dcq-computable states:
\begin{definition}
\label{def:kprep}
Let $\rho$ be a directly dcq-computable quantum state and $\mathsf{T}$ be a dcq-TM. The basic algorithmic complexity of $\rho$ is
\begin{equation}
\K^{d}_{\mathsf{T}}(\rho) = \min_p \left\{|p|: \mathsf{T}^{(Q)}(p;\epsilon) = \rho \right\}.
\end{equation}
\end{definition}

This definition is limited in that it is only applicable to directly dcq-computable states (denoted by the superscript $d$), which are dependent on the specific chosen universal set of gates $\qU$ in Def.~\ref{def:dcqtm}. We are interested in extending the concept of algorithmic complexity for dcq-computable states.

\begin{definition}
\label{def:kcomp}
Let $\rho$ be a quantum state and $\mathsf{T}$ be a dcq-TM. The algorithmic complexity of $\rho$ is given by
\begin{equation}
\label{eq:kcomp}
\K_{\mathsf{T}}(\rho) = \inf \left\{ |p|: \forall k \in \mathbb{N}
\:\: D \left( \mathsf{T}^{(Q)}\bigl( (p, k);\epsilon \right), \rho \Bigr) \leq \frac{1}{k}\right\}
\end{equation}
\end{definition}

 \begin{definition}
\label{def:kcomp_2}
Let $x$ be a binary string, $\rho, \sigma$ be two quantum states, and $\mathsf{T}$ be a dcq-TM. The conditional algorithmic complexity of $\rho$ given the pair $(x;\sigma)$ is
\begin{equation}
    \K_{\mathsf{T}}(\rho \:|\: x, \sigma) = \inf \left\{ |p|: \forall k \in \mathbb{N}
\:\: D\Bigl(\mathsf{T}^{(Q)} \left( (p, x , k); \sigma) \right), \rho \Bigr) \leq \frac{1}{k} \right\}.
\end{equation}
\end{definition}

In other words, we define the complexity of a quantum state as the infimum length of the {\em classical description} of the sequence that converges to such a state. Note that, because Equation~\eqref{eq:kcomp} considers the infimum of the set, states that are not dcq-computable are assigned infinite complexity. This definition is analogous to the standard classical complexity for real numbers (see Def.~\ref{def:kreals}) and it is machine independent in the following sense: 
\begin{property}
    \label{pr:constantoverhead}
    Let $\mathsf{T}$ and $\mathsf{T}'$ be universal dcq-TM, then for any string $x$ and quantum states $\rho, \sigma$ it holds that
    \begin{equation}
        \K_{\mathsf{T}}(\rho\:|\:x;\sigma) = \K_{\mathsf{T}'}(\rho\:|\:x;\sigma) + O(1)
    \end{equation}
\end{property}
\begin{proof}
    Let $s$ be a map that translates $\mathsf{T}'$-programs to $\mathsf{T}$-programs and $c\in \nats$ be such that $|s(p)| \leq |p| + c$ for every $p\in\{0,1\}^*$. Then, 
\begin{align}
    \K_{\mathsf{T}'}(\rho\:|\: x;\sigma) + c  &= \inf \left\{ |p|: \forall k \in \mathbb{N}
\:\: D\left(\mathsf{T'}^{(Q)} \left( (p, x , k); \sigma) \right), \rho \right) \leq \frac{1}{k} \right\} + c \nonumber \\
    &= \inf \left\{ |p|: \forall k \in \mathbb{N}
\:\: D\left(\mathsf{T}^{(Q)} \left( (s(p), x , k); \sigma) \right), \rho \right) \leq \frac{1}{k} \right\} + c \nonumber \\
    &\geq \inf \left\{ |s(p)|: \forall k \in \mathbb{N}
\:\: D\left(\mathsf{T}^{(Q)} \left( (s(p), x , k); \sigma) \right), \rho \right) \leq \frac{1}{k} \right\}  \nonumber \\
    &\geq \inf \left\{ |p|: \forall k \in \mathbb{N}
\:\: D\left(\mathsf{T}^{(Q)} \left( (p, x , k); \sigma) \right), \rho \right) \leq \frac{1}{k} \right\}  \nonumber \\
    &= \K_{\mathsf{T}}(\rho\:|\: x;\sigma).
\end{align}
By a symmetric argument, since $\mathsf{T}'$ is also assumed to enjoy the \smn\ property, there is $c'\in\nats$ such that:
\begin{equation}
\K_{\mathsf{T}}(\rho\:|\:x;\sigma) + c' \geq \K_{\mathsf{T}'}(\rho\:|\:x;\sigma).
\end{equation}
\end{proof}

\begin{remark}
   From Property~\ref{pr:constantoverhead} we can talk about the machine independent algorithmic complexity $\K$ (without subscript) and consider complexities that are equal up to an additive constant to be equivalent.
\end{remark}

We can turn our attention now into exploring the relationship between the complexity of a quantum state and its respective classical representation. As the following theorem states, they turn out to be same.
\begin{theorem}
    \label{thm:KCQequivalence}
    For every dcq-TM $\mathsf{T}$ and quantum state $\rho$ it holds that
    \begin{equation}
        \label{eq:KCQequivalence}
        \K(\mathbold{\rho}) = \K(\rho).
    \end{equation}
    \end{theorem}
    \begin{proof}
    The result follows directly from the algorithms in the proof of Theorem~\ref{thm:CQequivalence}, with the constant term being the length of each algorithm's specification.
    \end{proof}
An immediate implication of Theorem~\ref{thm:KCQequivalence} is that, since the components of a density matrix can have arbitrarily high complexity even for single qubit quantum systems, there is no upper bound for the complexity of an $n$-qubit quantum state. This is in contrast to the complexity of $n$-bit strings, which is upper bounded by $n$.\footnote{This is also a significant way in which this definition differs from Vitányi's approach, which is also based on classical descriptions, and for which the complexity of an $n$-bit pure state is upper bounded by $2n$ (see Theorem 3 from \cite{vit:01})}

Does Theorem~\ref{thm:KCQequivalence} mean that the information value of the state of a quantum system exactly coincides with that of the numbers in its density matrix? The reason the two quantities in Equation~(\ref{eq:KCQequivalence}) coincide is because we are quantifying the complexity of a quantum system using a classical string (the description of the program). Note that having access to the description of a program that computes a given quantum state is not the same as having a copy of the quantum state itself. To get a little better insight about how these two types of objects differ we need to consider their use in the role of a {\em resource}.

\subsection{Conditional complexity and state cloning}
Let us continue by evidencing a consequence of the no-cloning theorem. 
Since the machine receives a specification of the size of its quantum input, we know that for any state $\rho$ it holds that $\K(\rho \:|\: \blk;\rho)=0$, as the program that outputs the quantum input has constant length. On the other hand $\K(\rho \otimes \rho \:|\:\blk; \rho)\neq 0$, as there is not a single machine that can copy arbitrary quantum states. In the worst case scenario, such machine would construct the second copy of $\rho$ from scratch, which means that the complexity of duplicating a physical state is upper bounded by the complexity of the state itself:
\begin{align}
    \label{eq:computetoduplicate}
    \K(\rho \otimes \rho \:|\: \blk; \rho)\leq \K(\rho). 
\end{align}
This is in contrast to the case in which one would want to copy a classical input, in which case it is always possible to output twice the input. If one had \textit{a priori} access to the classical description of the state, it is possible to run Algorithm~\ref{alg:classicaltoquantum} repeatedly to compute several copies of $\rho$. Therefore, copying a quantum state given its associated density matrix $\mathbold{\rho}$ has constant complexity:
\begin{align}
\label{eq:freecopy}
 \K(\rho \otimes \rho \:|\: \mathbold{\rho}; \epsilon) = 0.
\end{align}
This exercise suggests that the classical description of a quantum state contains ``more" descriptive information that the physical state itself, and thus we can conclude that 
\begin{property}
\label{thm:condCQnonequivalence}
For any dcq-computable quantum states $\rho, \sigma$ it holds that
\begin{equation}
    \K(\rho \:|\: \mathbold{\sigma}; \epsilon) \leq \K(\rho \:|\: \blk; \sigma).
\end{equation}
\end{property}

We now turn our attention to answering if, for arbitrary quantum states, the complexity of copying the state is the same as computing it. In other words, we want to know if the converse of Equation~(\ref{eq:computetoduplicate}) holds. Intuitively, this would mean that any machine that computes the duplicate of a state $\rho$ would necessarily have the descriptive information of that state somehow encoded within it. In order to tackle this question, let us first state a couple of lemmas, the proofs of which can be found in the Appendix. Recall that the fidelity between two quantum states $\rho$ and $\sigma$ is defined as
\begin{equation}
    F(\rho, \sigma) = \left( \tr\sqrt{\sqrt{\rho} \sigma \sqrt{\rho}} \right)^2.
\end{equation}

\begin{lemma}
\label{lem:quasiduplicationdifelity}
Let $T:\mathcal{D}(\mathcal{H})\rightarrow \mathcal{D}(\mathcal{H}^{\otimes 2})$ be a CPTP map, i.e., a completely positive and trace-preserving map. Define the {\em copying error} $\mathtt{CE}_T$ of the transformation $T$ on the state $\rho \in \mathcal{D}(\mathcal{H})$ as
\begin{equation}
    \mathtt{CE}_{T}(\rho) = \sqrt{1-F(\rho^{\otimes 2}, T(\rho))}.
\end{equation}
For any $\rho_1, \rho_2 \in \mathcal{D}(\mathcal{H})$ such that $\mathtt{CE}_{T}(\rho_1) + \mathtt{CE}_{T}(\rho_2) \leq \varepsilon \leq \frac{1}{4}$ it holds that
\begin{equation}
    F(\rho_1, \rho_2) \leq \frac{1}{2} - \sqrt{\frac{1}{4} - \varepsilon} \;\; \; \textnormal{or} \;\;\; F(\rho_1, \rho_2) \geq \frac{1}{2} + \sqrt{\frac{1}{4} - \varepsilon}.
\end{equation}
\end{lemma}

\begin{lemma}
\label{lem:maxquasiduplication}
Let $\mathcal{H}$ be an $N$-dimensional Hilbert space and $A = \lbrace \rho_i \in \mathcal{D}(\mathcal{H})\rbrace$ be a set of density operators such that for $i \neq j$:
\begin{equation}
    \label{eq:maxquasiduplication}
    F(\rho_i, \rho_j) \leq \delta_0(N),
\end{equation}
with
\begin{equation}
    \delta_0(N) =\frac{1}{N^2}\left( \frac{1-\sqrt{1-\frac{1}{N}}}{1+2^{\left(\frac{3^N+1}{2}\right)}} \right)^4.
\end{equation}
Then, $A$ cannot have more than $N$ elements. 
\end{lemma}

\begin{lemma}
\label{lem:smallproperties}
For any CPTP map $T:\mathcal{D}(\mathcal{H}) \rightarrow \mathcal{D}\left(\mathcal{H}^{\otimes 2}\right)$:
\begin{enumerate}
    \item  $D\left(T(\rho), \rho^{\otimes 2}\right) \leq {1}/{k} \implies \mathtt{CE}_{T}(\rho) \leq \sqrt{{2}/{k}}$.
    \item $D\left(T(\rho), \rho^{\otimes 2}\right) \leq 3D(\rho, \rho') + D\left(T(\rho'), \rho'^{\otimes 2}\right)$.
\end{enumerate}
\end{lemma}

\begin{theorem}
\label{thm:qcloning}
For any $n$-qubit quantum state $\rho$ it holds that
\begin{equation}
    \label{eq:duplicatetocompute}
    \K(\rho) \leq \K(\rho \otimes \rho \:|\: \blk; \rho)  + 2n + 2\log(n).
\end{equation}
\end{theorem}
\begin{proof}
To prove the result, we construct a family of algorithms (See Algorithm.~\ref{alg:constructfromcopy}) with three parameters: a program index $t$, a qubit number $n$, and a list index $m$. Then show that, given an $n$-qubit state $\rho$, there exist $t,m \in \mathbb{N}$ satisfying:
\begin{equation}
    |t| \leq K(\rho \otimes \rho | \blk; \rho) + c; \;\;\;\;\;\;\;\;\;\; |m| \leq n,
\end{equation}
where $c$ is a constant independent of $n$, such that the program specified by Algorithm~\ref{alg:constructfromcopy} with the parameters $(t,n,m)$ computes $\rho$. We can see Algorithm~\ref{alg:constructfromcopy} as a program that computes $\rho$ from $\epsilon$ using a program that computes $\rho \otimes \rho$ from $\rho$ as a subroutine.
\begin{center}
    \begin{algorithm}[H]
    \caption{Approximation of an $n$-qubit state given a subroutine that is known to be able to duplicate it and a specification index $m \leq \log(n)$.}
    \label{alg:constructfromcopy}
    \textbf{Parameters:} Universal dcq-TM $\mathsf{U}$. Integers $t,n,m$\; 
    \textbf{Input:} Integer $k$\; 
    \vspace{3mm}
    (List preparation phase)\\
    $L \leftarrow \left( \right); s \leftarrow 0; \varepsilon_0 \leftarrow \frac{1}{4}\delta_{0}(2^n)(1-\delta_{0}(2^n)); k_0 \leftarrow\left(\frac{4}{\varepsilon_0}\right)^2 $\;
    \While{The size of the list $|L| \leq m$}{
        $\mathbold{\sigma} \leftarrow \Tilde{\Pi}_1(s)$\;
        \If{$\mathbold{\sigma}$ is a $2^n \times 2^n$ matrix such that $\mathtt{CE}_{T_{t, k_0}}(\mathbold{\sigma}) \leq \frac{2}{\sqrt{k_0}} \And F(\mathbold{\sigma}, \mathbold{\sigma}') \leq \delta_0$ for all $\mathbold{\sigma}' \in L$}
        {Add $\mathbold{\sigma}$ to $L$;}
        $s \leftarrow s+1$\;
    }
    $\mathbold{\rho}_0 \leftarrow L[m]$  (the $m$-th element of the list $L$)\; 
    
    \vspace{3mm}
    (Output preparation phase)\\
    $s \leftarrow 0; \varepsilon_{\textnormal{out}} \leftarrow \min \left\{ \delta_{0}(2^n)(1-\delta_{0}(2^n)), \frac{1}{k^2}\left(1-\frac{1}{k^2}\right) \right\}; k_{\textnormal{out}} \leftarrow\left (\frac{4}{\varepsilon_{\textnormal{out}}}\right)^2 $\;
    \While{no matrix $\mathbold{\rho}_{\textnormal{out}}$ is found}{
        $\mathbold{\sigma} \leftarrow \Tilde{\Pi}_1(s)$\;
        \eIf{$\mathbold{\sigma}$ is a $2^n \times 2^n$ matrix such that $\mathtt{CE}_{T_{t, k_{\textnormal{out}}}}(\mathbold{\sigma}) \leq \frac{2}{\sqrt{k_{\textnormal{out}}}} \And F(\mathbold{\sigma}, \mathbold{\rho}_0) > \frac{16}{\sqrt{k_{0}}}$}
        {$\mathbold{\rho}_{\textnormal{out}} \leftarrow \mathbold{\sigma}$;}
        {$s \leftarrow s+1$;}
    }
    \Return{$\left(\blk; \rho_{\textnormal{out}} = \Pi_1(s) \right)$};
    \end{algorithm}
\end{center}
Let $\mathsf{U}$ be a reference universal dcq-TM. For any $t, k \in \mathbb{N}$ define the family of CPTP maps $\left\lbrace T_{t,k}: \mathbb{C}^{2^n} \rightarrow  \mathbb{C}^{2^{2n}} \right\rbrace$ as: 
\begin{equation}
    T_{t,k}(\sigma) = \mathsf{U}^{Q}_{t}(k;\sigma).
\end{equation}
That is, $T_{t,k}$ represents the transformation (understood as going from the quantum input space to the quantum output space) applied by the machine $\mathsf{U}$ running the program $t$ on classical input $k$.

Let us now gain a little intuition on the two phases of Algorithm~\ref{alg:constructfromcopy}. The list preparation phase goes through the set of directly dcq-computable states looking for states that are ``close'' to being copied by $T_{t, k_0}$. Whenever it finds one such state, it adds it to the list $L$ if it is ``far'' from all the other states already in the list. The phase ends when the $m$-th element is added to the list this way and assigns $\mathbold{\rho}_0$ to be equal to the last state added.  On the other hand, the output preparation phase goes again through the set of directly dcq-computable states, now looking for a state that is both close to the state $\rho_0$ obtained from the first phase and ``close enough" to being copied by $T_{t, k_{\textnormal{out}}}$. To show that the proposed algorithm satisfies the desired properties, we break down the proof into two parts. We want to show that:
\begin{enumerate}[i)]
    \item If $\mathsf{U}_t$ computes $\rho \otimes \rho$ from $\rho$ and for large enough $m$, the list preparation phase will always add a matrix $\mathbold{\rho}_0$ for which 
    \begin{equation}
    \label{eq:listphasecondition}
        F(\mathbold{\rho}_0, \mathbold{\rho}) \geq 1-\delta_0(2^n),
    \end{equation}
    and such matrix will be found before the list exceeds $2^n$ elements
    \item On input $k$, if $\mathbold{\rho}_0$ satisfies Equation~(\ref{eq:listphasecondition}) then the output state $\rho_{\textnormal{out}}(k)$ will satisfy
    \begin{equation}
    \label{eq:outputphasecondition}
        D(\rho_{\textnormal{out}}(k), \rho) \leq \frac{1}{k}.
    \end{equation}
\end{enumerate}
Once these properties are proven, the bound in Equation~(\ref{eq:duplicatetocompute}) follows by noting that
the minimal size program $t'$ that computes $\rho \otimes \rho$ from $\rho$ can be specified by using $|t'| = K_{\mathsf{U}}(\rho \otimes \rho | \blk; \rho) \leq K(\rho \otimes \rho | \blk; \rho) + c$ bits, and the values of $m \leq 2^{n}$ and $n$ can be specified with $n$ and $\log(n)$ bits, respectively. 
\begin{itemize}
    \item Proof of statement i)
\end{itemize}
Note that whenever a matrix $\mathbold{\sigma}$ gets added to $L$ it satisfies by Lemma~\ref{lem:smallproperties}(1) and the definition of dcq-computability
\begin{align}
    \label{eq:sumCEk0}
    \mathtt{CE}_{T_{t, k_0}}(\mathbold{\sigma}) + \mathtt{CE}_{T_{t, k_0}}(\mathbold{\rho}) \leq \frac{2}{\sqrt{k_0}} + \sqrt{\frac{2}{k_0}} \leq \frac{4}{\sqrt{k_0}},
\end{align}
invoking Lemma~\ref{lem:quasiduplicationdifelity} we get that
\begin{equation}
    \label{eq:fidelitiesk0}
    F(\mathbold{\sigma}, \mathbold{\rho}) \leq \frac{1}{2} - \sqrt{\frac{1}{4} - \frac{4}{\sqrt{k_{0}}}} \;\; \; \textnormal{or} \;\;\; F(\mathbold{\sigma}, \mathbold{\rho}) \geq \frac{1}{2} + \sqrt{\frac{1}{4} - \frac{4}{\sqrt{k_{0}}}},
\end{equation}
and substituting $k_0$ in terms of $\delta_0(2^n)$
\begin{equation}
    \label{eq:fidelitiesdelta0}
    F(\mathbold{\sigma}, \mathbold{\rho}) \leq \frac{\delta_0}{4} \;\; \; \textnormal{or} \;\;\; F(\mathbold{\sigma}, \mathbold{\rho}) \geq 1-\frac{\delta_0}{4}.
\end{equation}
Depending on the transformation $T_{t,k_0}$, the loop inside the \textbf{while} statement of the list preparation phase will add some amount of elements to the list before there are no more states that satisfy the stated conditions. Let $l_{\textnormal{max}}$ be the maximum amount of elements that can be added to $L$ in this way. From Lemma~\ref{lem:maxquasiduplication} we know that the $l_{\textnormal{max}} \leq 2^n$ because the largest size of a set $\lbrace \mathbold{\sigma}_i \rbrace$ that satisfies $F(\mathbold{\sigma}_i, \mathbold{\sigma}_j) \leq \delta_0(\mathbb{C}^{2^n})$ for $i \neq j$ is equal to $\textnormal{dim}(\mathbb{C}^{2^n}) = 2^n$. Recall now that the set of $n$-qubit directly dcq-computable states is dense in the set of $n$-qubit quantum states and the fact that the function $\tilde{\Pi}_1(s)$ will eventually output every directly dcq-computable state. Let $s'$ be the smallest integer such that
\begin{equation}
    \label{eq:closepreparable}
    D( \mathbold{\rho}' = \tilde{\Pi}_1(s'), \mathbold{\rho}) \leq \frac{1}{3 k_0},
\end{equation}
using Lemma~\ref{lem:smallproperties}(2) and the fact that $\mathsf{U}_t$ computes $\rho \otimes \rho$ from $\rho$
\begin{align}
    D\left(T_{t, k_0}(\mathbold{\rho}'), \mathbold{\rho}'^{\otimes 2}\right) &\leq 3D(\mathbold{\rho}', \mathbold{\rho}) + D\left(T_{t, k_0}(\mathbold{\rho}), \mathbold{\rho}^{\otimes 2}\right) \\
    &\leq \frac{2}{k_0},
\end{align}
and from Lemma~\ref{lem:smallproperties}(1) we get
\begin{equation}
    \label{eq:guaranteedmatrix}
    \mathtt{CE}_{T_{t, k_0}}(\mathbold{\rho}') \leq \frac{2}{\sqrt{k_0}}.
\end{equation}
Suppose the loop reaches $s'$ without finding $\mathbold{\sigma}$ satisfying Equation~(\ref{eq:listphasecondition}). Then, as shown in Equation~(\ref{eq:guaranteedmatrix}), at $s=s'$ the matrix $\mathbold{\rho}' = \tilde{\Pi}_1(s')$ will satisfy the first condition in the check to be added to $L$. Additionally, for every $\mathbold{\sigma}$ already in $L$ at that point, it will hold that (see Lemma 1 from~\cite{ras:03})
\begin{equation}
    \sqrt{1-F(\mathbold{\rho}', \mathbold{\sigma})} \geq |\underbrace{F(\mathbold{\rho}', \mathbold{\rho})}_{\geq 1 - \frac{\delta_0}{4}} - \underbrace{F(\mathbold{\rho}, \mathbold{\sigma})}_{\leq \frac{\delta_0}{4}} | \geq 1 - \frac{\delta_0}{2},
\end{equation}
and hence
\begin{equation}
    F(\mathbold{\rho}', \mathbold{\sigma}) \leq \delta_0 \left(1 - \frac{\delta_0}{4} \right) \leq \delta_0,
\end{equation}
which satisfies the second condition to be added to $L$. Therefore, for large enough $m$, the list $L$ will always have an element that satisfies Equation~(\ref{eq:listphasecondition}). Particularly, there is a value of $m \leq l_{\textnormal{max}} \leq 2^n$ for which the cycle ends when such element is added and sets it to be $\mathbold{\rho}_0$.
\begin{itemize}
    \item Proof of statement ii)
\end{itemize}
Suppose $\mathbold{\rho}_0$ satisfies Equation~(\ref{eq:listphasecondition}). We proceed to show that any density matrix $\mathbold{\sigma}$ that holds both of the conditions in the \textbf{if} statement of the output preparation phase must also satisfy
\begin{equation}
    D(\mathbold{\sigma}, \mathbold{\rho}) \leq \frac{1}{k}.
\end{equation}
Analogously to Equations~(\ref{eq:sumCEk0}) and~(\ref{eq:fidelitiesk0}), whenever a matrix $\mathbold{\sigma}$ is found that satisfies $\mathtt{CE}_{T_{t, k_{\textnormal{out}}}}(\mathbold{\sigma}) \leq {2}/{\sqrt{k_{\textnormal{out}}}}$ we have that
\begin{align}
    \label{eq:sumCEkout}
    \mathtt{CE}_{T_{t, k_{\textnormal{out}}}}(\mathbold{\sigma}) + \mathtt{CE}_{T_{t, k_{\textnormal{out}}}}(\mathbold{\rho}) \leq \frac{2}{\sqrt{k_{\textnormal{out}}}} + \sqrt{\frac{2}{k_{\textnormal{out}}}} \leq \frac{4}{\sqrt{k_{\textnormal{out}}}},
\end{align}
and
\begin{equation}
    \label{eq:fidelitieskout}
    F(\mathbold{\sigma}, \mathbold{\rho}) \leq \frac{1}{2} - \sqrt{\frac{1}{4} - \frac{4}{\sqrt{k_{\textnormal{out}}}}} \;\; \; \textnormal{or} \;\;\; F(\mathbold{\sigma}, \mathbold{\rho}) \geq \frac{1}{2} + \sqrt{\frac{1}{4} - \frac{4}{\sqrt{k_{\textnormal{out}}}}}.
\end{equation}
In other words, such $\mathbold{\sigma}$ is either ``very close" to $\mathbold{\rho}$ or ``very far" from it. Consider the case where left part of Eq~(\ref{eq:fidelitieskout}) is true, we have that (noting that $k_{\textnormal{out}} \geq k_0 $)
\begin{align}
    \sqrt{1-F(\mathbold{\rho}_0, \mathbold{\sigma})} &\geq | \underbrace{F(\mathbold{\rho}_0, \mathbold{\rho})}_{\substack{\geq \\ \frac{1}{2} + \sqrt{\frac{1}{4} - \frac{4}{\sqrt{k_{0}}}}}} - \underbrace{F(\mathbold{\rho}, \mathbold{\sigma})}_{\substack{\leq \\ \frac{1}{2} - \sqrt{\frac{1}{4} - \frac{4}{\sqrt{k_{\textnormal{out}}}}}}}| \nonumber \\
    &\geq \left|\sqrt{\frac{1}{4} - \frac{4}{\sqrt{k_{0}}}} + \sqrt{\frac{1}{4} - \frac{4}{\sqrt{k_{\textnormal{out}}}}} \right| \nonumber \\
    &\geq \sqrt{1 - \frac{16}{\sqrt{k_0}}},
\end{align}
and therefore
\begin{equation}
    F(\mathbold{\rho}_0, \mathbold{\sigma}) \leq \frac{16}{\sqrt{k_0}},
\end{equation}
which means that such $\mathbold{\sigma}$ cannot satisfy the second condition of the \textbf{if} statement. We conclude that for any $\mathbold{\sigma}$ that satisfies both conditions (and hence also for $\mathbold{\rho}_{\textnormal{out}}$), it holds that
\begin{align}
     F(\mathbold{\sigma}, \mathbold{\rho}) &\geq \frac{1}{2} + \sqrt{\frac{1}{4} - \frac{4}{\sqrt{k_{\textnormal{out}}}}} \nonumber \\
     &\geq \frac{1}{2} + \sqrt{\frac{1}{4} - \frac{1}{k^2}\left(1-\frac{1}{k^2}\right)} \nonumber \\
     &= 1 - \frac{1}{k^2}.
\end{align}
Finally, translating into trace distance we get
\begin{equation}
    D(\mathbold{\sigma}, \mathbold{\rho}) \leq \sqrt{1-F(\mathbold{\sigma}, \mathbold{\rho})} \leq \frac{1}{k}
\end{equation}
\end{proof}

One way to interpret Theorem~\ref{thm:qcloning} along with Equation~\eqref{eq:freecopy} is that, for the task of describing $\rho \otimes \rho$, the advantage of having prior access to a single copy of $\rho$ as compared to not having any prior resource state is at most $2n +2\log(n)$ bits (up to an additive constant). This means that for complex quantum states, for which $\K(\rho) \gg n$, having access to a copy of $\rho$ gives effectively \emph{no advantage} for describing $\rho \otimes \rho$. This is in contrast with classical strings, where having a copy of the string gives all the information needed to compute its duplicate, and can be understood as an algorithmic information version of the no-cloning theorem.


\section{On the chain rule and complexity of quantum correlations}
\label{sec:kcorrelations}

One feature of quantum multipartite states is that, in general, they contain correlations between their parts. One way to quantify correlations is through the notion of mutual information. In this section we use the dcq-TM model to generalize the chain rule and extend the concept of algorithmic mutual information to quantum states. Let us first define the concept of prefix for quantum states.
\begin{definition}(Prefix of a quantum state) \label{def:qprefix} \\
Let $\rho_{\langle 1,n\rangle} \in \mathcal{D}\left(\mathbb{C}^{2^{n}}\right)$. For $m \leq n$, we say that $\sigma \in \mathcal{D}(\mathbb{C}^{2^{m}})$ is a prefix of $\rho_{\langle 1,n\rangle}$ whenever
\begin{equation}
    \sigma = \textnormal{Tr}_{(m+1, \ldots, n)}\left[\rho_{\langle 1,n\rangle}\right].
\end{equation}
\end{definition}
We can make use of quantum prefixes to study the behaviour of $\K$ when considering systems split into two (or more) parts. Note that an $n$-qubit state has $n$ possible prefixes, which are manifestly defined by the layout of the cells in the quantum tape. When we want to stress the number $m$ of qubits of a prefix we will refer to it as the ``$m$-prefix". To make the separation more explicit in some of the following definitions we will denote the larger $n$-qubit space as $\mathcal{H}_{AB}$ and the $m$-qubit prefix space as $\mathcal{H}_{A}$ so we can write a quantum state and its prefix as $\rho_{AB}$ and $\rho_{A}$, respectively. 

By analogy with the classical algorithmic information theory, it seems natural to begin our exploration of mutual information by considering the straightforward generalization of two (classically) related quantities:

\begin{definition}(Algorithmic mutual information 1) \label{def:mi1}\\
Let $\rho_A$ and $\rho_B$ be two quantum states. Define the mutual information of $\rho_A$ respect to $\rho_B$ by:
\begin{equation}
\label{eq:mi1}
\I^{(1)}(\rho_A : \rho_B) = \K(\rho_A) - \K(\rho_A \: \vert \: \mathbold{\rho}_B).
\end{equation}
\end{definition}

\begin{definition}(Algorithmic mutual information 2) \label{def:mi2}\\
Let $\rho_{AB}$ be a bipartite quantum state. Define the mutual information between the partial states $\rho_A$ and $\rho_B$ by:
\begin{equation}
\label{eq:mi2}
\I^{(2)}(\rho_A : \rho_B) = \K(\rho_{A}) + \K(\rho_B ) - \K(\rho_{AB}).
\end{equation}
\end{definition}

Definition~\ref{def:mi1} compares the algorithmic complexity of the state $\rho_A$ assuming the machine starts with the all zero state in its classical and quantum tapes, versus its complexity assuming the machine has the description of the state $\rho_B$ given as a classical input. The choice of giving the classical description instead the state itself will become clearer as we study its potential in the context of correlations. Evidently $\K(\rho_A \vert \mathbold{\rho}_B)$ is upper bounded by $\K(\rho_{A})$, and hence $\I^{(1)}$ is non-negative. On the other hand, Def.~\ref{def:mi2} considers the difference between computing the states $\rho_A$ and $\rho_B$ separately as opposed to computing the (potentially correlated) single state $\rho_{AB}$. The main conceptual difference between $\I^{(1)}$ and $\I^{(2)}$ is whether to consider the states $\rho_A$ and $\rho_B$ to exist independently or as a part of a greater joint state. Under the classical definition of Kolmogorov complexity, the analogous expressions to (\ref{eq:mi1}) and (\ref{eq:mi2}) for mutual information turn out to be connected through the chain rule as shown by Equation~\eqref{eq:chainrulek2}. To explore what is the connection between these two quantities, we may ask ourselves how would a quantum version of the chain rule look like? A straightforward generalization of Equation~\eqref{eq:chainrulek2} would be to replace the string $xy$ with a bipartite state $\rho_{AB}$ leading to an expression of the form:
\begin{equation}
    \label{eq:incorrectchainrule}
     \K(\rho_{AB}) \underset{?}{=} \K(\rho_{A}) + \K(\rho_{B} \mid \mathbold{\rho}_{A}) + O \bigl(\textnormal{log} (n \K(\rho_{AB}))\bigr),
\end{equation}
but as Theorem~\ref{thm:noqchainrule} states, the above equality does not hold for arbitrary joint states.

\begin{theorem}
\label{thm:noqchainrule}
There exist no constant $g$ such that
\begin{equation}
    \label{eq:qchainrule}
    \K(\rho_{AB}) \leq \K(\rho_{A}) + \K(\rho_{B} \mid \mathbold{\rho}_{A}) + g \ \textnormal{log} (n\K(\rho_{AB}))
\end{equation}
for every natural number $n$ and $n$-qubit quantum bipartite state $\rho_{AB}$.
\end{theorem}

\begin{proof}
Let $g \in \mathds{R}$, $\rho_{AB} \in \mathcal{D}((\mathds{C}^2)^{\otimes n = n_a +n_b})$ such that $\rho_{A} \in \mathcal{D}((\mathds{C}^2)^{\otimes n_a})$ and $\rho_{B} \in \mathcal{D}((\mathds{C}^2)^{\otimes n_b})$. Let $\rho_A = \sum_{i=1}^{2^{n_a}} \alpha_i \ket{i}\!\bra{i}$ and $\rho_B = \sum_{j=1}^{2^{n_b}} \beta_i \ket{j}\!\bra{j}$ be the spectral decomposition of the reduced systems $\rho_A$ and $\rho_B$. Assume that the Inequality~(\ref{eq:qchainrule}) holds for some $\rho_{AB}$ for which $\rho_A$ and $\rho_B$ are proper mixed states and some constant $g$. We show that there exists then a $\Tilde{\rho}_{AB}$ for which the inequality does not hold. Hence, there is no $g$ that satisfies Equation~(\ref{eq:qchainrule}) for all bipartite mixed states.

Note that since $\rho_A$ and $\rho_B$ are proper mixed states, there exist $r,r'$ and $s,s'$ such that $r \neq r'$, $s \neq s'$, and $\alpha_{r},\alpha_{r'},\beta{s},\beta{s'}$ are all non-zero. Define the set of states
\begin{equation}
    \Tilde{\mathcal{C}} = \left\lbrace \Tilde{\rho}_{AB}^{\lambda} \mid \Tilde{\rho}_{AB}^{\lambda} = \lambda(\rho_A \otimes \rho_B) + (1-\lambda)\sigma_{AB}, \: \textnormal{such that} \:  \lambda \in (0,1) \: \textnormal{is computable} \: \right\rbrace ,
\end{equation}
where $\rho_A \otimes \rho_B = \sum_{i,j=1}^{2^{n_a},2^{n_b}} \alpha_i \beta_j \ket{i,j}\!\bra{i,j}$, and 
\begin{equation}
    \sigma_{AB} = \rho_A \otimes \rho_B + \gamma(\ket{r, s}\bra{r',s'} + \ket{r',s'}\!\bra{r,s}),
\end{equation}
where $\gamma =\sqrt{\alpha_{r}\alpha_{r'}\beta_{s}\beta_{s'}}$. Note that $\sigma_{AB}$ and $\rho_A \otimes \rho_B$ are strictly different density operators and, since the partial traces on the second term of $\sigma_{AB}$ both vanish, we have that for any $0 \leq \lambda \leq 1$
\begin{align}
    \Tilde{\rho}_{A}^{\lambda}
    &= \textnormal{Tr}_{B}\Tilde{\rho}_{AB}^{\lambda} \nonumber \\
    &= \lambda \textnormal{Tr}_{B}[\rho_A \otimes \rho_B] + (1-\lambda) \textnormal{Tr}_{B}\sigma_{AB} \nonumber \\
    &= \lambda \textnormal{Tr}_{B} [\rho_A \otimes \rho_B] + (1-\lambda)\left(\textnormal{Tr}_{B} [\rho_A \otimes \rho_B] + \gamma \textnormal{Tr}_{B} \left[\ket{r, s}\!\bra{r',s'} + \ket{r',s'}\!\bra{r,s}\right]\right) \nonumber \\
    &= \rho_A,
\end{align}
and similarly $\Tilde{\rho}_{B}^{\lambda} = \rho_B$. Notice also that, since $\rho_A$ and $\rho_B$ are assumed to be dcq-computable, it follows from Theorem~\ref{thm:CQequivalence} all of the $ \Tilde{\rho}_{AB}^{\lambda}$ are dcq-computable quantum states. Let us denote the value $g_0 = \K(\rho_{A}) + \K(\rho_{B} \mid \mathbold{\rho}_{A})$. We can write the inequality (\ref{eq:qchainrule}) as:
\begin{equation}
    \K(\rho_{AB}) - g \textnormal{log}(n \K(\rho_{AB})) - g_0 \leq 0.
\end{equation}
Define now the set
\begin{equation}
    \mathcal{C}_m = \lbrace \rho_{AB}' \mid \rho_{AB}' \in \mathcal{D}((\mathds{C}^2)^{\otimes n}),
     \K(\rho_{AB}') \leq m \rbrace.
\end{equation}
Note that for all $m$, the set $\mathcal{C}_m$ is finite (of size at most $2^m$), whereas the set $\Tilde{\mathcal{C}}$ is infinite (its size is equal to the set of computable real numbers between 0 an 1), this means that for any $m$, there exists $\lambda$ such that $\Tilde{\rho}_{AB}^{\lambda} \notin \mathcal{C}_m$ ($K(\Tilde{\rho}_{AB}^{\lambda}) \geq m$). In other words, given a bipartite state $\rho_{AB}$ with mixed partial traces, there are an infinite number of states $\rho_{AB}'$ that share its partial traces and have larger quantum Kolmogorov complexity under the dcq-TM model. Additionally, for any $n,g,g_0$ we have that
\begin{equation}
    \lim_{m\to\infty} \bigl(m - g \textnormal{log}(nm) - g_0 \bigr) = \infty.
\end{equation}
Therefore, there is an $m_0$ such that $m' - g \textnormal{log}(nm') - g_0 > 0$ for any $m' \geq m_0$. The result follows by noting that, from the above argument, there exists $\lambda$ such that $\K(\Tilde{\rho}_{AB}^{\lambda}) \geq m_0$, and hence
\begin{equation}
    \label{eq:unboundchainrule}
    \K(\Tilde{\rho}_{AB}^{\lambda}) > \K(\Tilde{\rho}_{A}^{\lambda}) + \K(\Tilde{\rho}_{B}^{\lambda} \mid \Tilde{\mathbold{\rho}}_{A}^{\lambda}) + g \ \textnormal{log} (n\K(\Tilde{\rho}_{AB}^{\lambda})).
\end{equation}
\end{proof}

To understand better relation between Def.~\ref{def:mi1} and Def~\ref{def:mi2}, we have to consider the role of correlations in the description of bipartite states. In the Shannon approach, the notion of information is based on the knowledge of system's state: the more we are uncertain of the state of the system, the more we increase our knowledge (i.e., information) upon learning that state. Thus, mutual information between two systems is manifestly linked to correlations between them. In the original Kolmogorov approach, the notion of information is based on our ability to compute a given string. Consequently, mutual information between two strings tells us how much knowing one helps in computing the other. When generalizing to the quantum domain, we see that part of this intuition holds for Def.~\ref{def:mi1} as well. In other words, algorithmic mutual information is quantifying the similarity between the descriptions of two states, complexity wise, and says noting about the correlations in their joint state. Indeed, the quantities on the right hand side of Equation~(\ref{eq:mi1}) do not depend on the joint state $\rho_{AB}$ at all. On the other hand, Def.~\ref{def:mi2} effectively compares the complexity of computing a (generally) correlated state with its individual parts. One problem with Def.~\ref{def:mi2} as a measure of complexity is that, following Theorem~\ref{thm:noqchainrule}, the quantity is unbounded from below. 

At this point we can ask ourselves: is there a version of the chain rule that holds for quantum states? Such expression would help us motivate an alternative definition for algorithmic mutual information that satisfies a form of symmetry, as well as give us insight into the complexity of correlations in quantum systems. Theorem~\ref{thm:yesqchainrule} provides a relatively straightforward generalization of Equation~\eqref{eq:chainrulek2} by introducing a direct dependence on the joint state, as well as requiring a classical description of the prefix (as opposed to a copy of the state).

\begin{theorem}
\label{thm:yesqchainrule}
Let $n,m$ be integers such that $1< m \leq n$, and $\rho$ be an $n$-qubit quantum state. For any $m$-prefix $\sigma$ of $\rho$ it holds that
\begin{equation}
    \label{eq:yesqchainrule}
    \K(\rho) = \K(\sigma) + \K(\rho | \mathbold{\sigma}) + O\left(\log(n\K(\rho))\right).
\end{equation}
\end{theorem}
\begin{proof}
Let $\mathsf{U}$ be a reference universal dcq-TM:  \\
($\leq$): \\
We specify an algorithm to compute $\rho$ given as input two parameters: the specifications $t_1, t_2$ of two programs that compute $\sigma$ and $\rho | \mathbold{\sigma}$, respectively. The algorithm works as follows: on input $k$, simulates the action of $\mathsf{\Tilde{U}}_{t_1}(2k;\varepsilon)$ to obtain $\mathbold{\sigma_{k}}$, then simulate $\mathsf{U}_{t_2}((2k, \mathbold{\sigma_{k}});\varepsilon)$ to obtain $\rho'_{k}$. Let $\rho_{k} = \mathsf{U}_{t_2}((2k, \mathbold{\sigma});\varepsilon)$, using the triangle inequality and the contractive property of the trace distance we get that 
\begin{align}
    D(\rho, \rho'_k) &\leq D(\rho, \rho_{k}) + D(\rho_{k}, \rho'_{k}) \nonumber \\
    &\leq D(\rho, \rho_{k}) + D(\mathbold{\sigma}, \mathbold{\sigma_{k}}) \nonumber \\
    &\leq \frac{1}{2k} +  \frac{1}{2k} = \frac{1}{k}. 
\end{align}
So the algorithm correctly computes $\rho$. Specifying $t_1, t_2$ can be done using $\K(\sigma) + \K(\rho | \mathbold{\sigma})$ bits. Since this is a two-input program, a specification of the length of one of the inputs must be added, which is of size at least $\min \lbrace \log(\K(\sigma)), \log(\K(\rho | \mathbold{\sigma})) \rbrace \leq \log(\K(\rho)) + \log(m)$. We can conclude then that
\begin{equation}
    \K(\rho) \leq \K(\sigma) + \K(\rho | \mathbold{\sigma}) + \log(m\K(\rho)) \leq \K(\sigma) + \K(\rho | \mathbold{\sigma}) + \log(n\K(\rho))
\end{equation}
($\geq$): \\
By {\em reductio ad absurdum}, assume that for every positive real number $c$ there exist an integer $n > 1$, an $n$-qubit quantum state $\rho$, and a prefix $\sigma$ of $\rho$ such that
\begin{equation}
    \label{eq:chainassumption}
    \K(\rho) < \K(\sigma) + \K(\rho | \mathbold{\sigma}) - c \log(n \K(\rho)).
\end{equation}
Let $(c, \rho, \sigma)$ be such that they satisfy Equation~\eqref{eq:chainassumption} with $\rho, \sigma$ being $n$ and $m \leq n$ qubit states, respectively. Furthermore, let $t_{\rho}$ be such that $\mathsf{U}_{t_{\rho}}$ computes $\rho$ and $|t_{\rho}| = \K(\rho) + O(1)$. In the following treatment, it will be useful to have a shorthand notation for the prefixes of the quantum output of programs. Let $\eta(t,x)$ denote the $m$-prefix of $\mathsf{U}_{t}^{Q}(x;\varepsilon)$. For any positive integer $s$ consider now the following sets:
\begin{equation}
    P_{s} = \left\{ t \in \lbrace 0, 2^{|t_{\rho}|} \rbrace  | \; \textnormal{for all} \; 1 \leq x \leq s, \; \mathsf{U}_{t}(x; \varepsilon) \; \textnormal{halts with a quantum output of at least} \; m \; \textnormal{qubits}\right\}
\end{equation}
\begin{equation}
    \label{eq:pseta}
    P_{s}^{\eta} = \left \lbrace t \in P_{s}  | \; \textnormal{for all} \; 1 \leq x \leq s, \; D(\eta(t,x), \eta)\leq \frac{1}{x} \right \rbrace
\end{equation}
Recall that a set $A$ is computably enumerable if there exists a bijective function $\mathcal{E}_{A}: \lbrace 1, \ldots, |A| \rbrace \rightarrow A$ and a machine $\mathsf{T}_{A}$ such that $\mathsf{T}_{A}(i;\varepsilon) = (\mathcal{E}_{A}(i); \varepsilon)$ for every $i = \lbrace 1, \ldots, |A| \rbrace$. To see that the set $P_{s}^{\eta}$ is computably enumerable given $s, m, |t_{\rho}|$, and $\mathbold{\eta}$, consider a machine that classically simulates instances of all programs up to length $|t_{\rho}|$ running with all inputs $x \in \lbrace 1,\ldots, s \rbrace$ in parallel, keeping track of their simulated outputs as they halt, and outputting the $i$-th program found to satisfy the conditions in Equation~\eqref{eq:pseta}. Denote the enumeration of $P_{s}^{\eta}$ defined by this machine as $\mathcal{E}_{s,m,|t_{\rho}|,\mathbold{\eta}}$.

We want to use the set $P_{s}^{\sigma}$ to find an upper bound for $\K(\rho| \mathbold{\sigma})$. Because we know that $t_{\rho} \in P_{s}^{\sigma}$ for any $s$, fix $s=s_0$ and let $j$ be the index assigned to $t_{\rho}$ by $\mathcal{E}_{s_{0},m,|t_{\rho}|,\mathbold{\sigma}}$. We can compute $\rho| \mathbold{\sigma}$ as follows: consider a machine that on input $(k;\mathbold{\sigma})$ outputs the values of $|t_{\rho}|$, $m$, and $j$, then finds the value of $t_{\rho} = \mathcal{E}_{s_0,m,|t_{\rho}|,\mathbold{\sigma}}(j)$ and outputs the simulated quantum output of $\mathsf{U}_{t_{\rho}}(k)$. By recalling that $|m| \leq \log(n)$ we get that
\begin{equation}
    \label{eq:boundkrhosigma}
    \K(\rho| \mathbold{\sigma}) \leq \log(|P_{s_0}^{\sigma}|) + 2\log(\K(\rho)) + 2\log(n) = \log(|P_{s_0}^{\sigma}|) + 2\log(n \K(\rho)).
\end{equation}
Using Equation~\eqref{eq:chainassumption} we can obtain a bound for the size of $P_{s_0}^{\sigma}$
\begin{equation}
    \label{eq:boundpssigma}
    \log(|P_{s_0}^{\sigma}|) > \K(\rho) - \K(\sigma) + (c-2)\log(n \K(\rho)) = \ell 
\end{equation}
which is independent of the value of $s_0$. To arrive to a contradiction, we will proceed to show that the bound found in Equation~\eqref{eq:boundpssigma} is ``too big'', in the sense that it allows for a very short description of $\sigma$. To do this, we will use the fact that we can compute $\sigma$ by simulating a program that computes $\rho$ and then outputting only the first $m$ qubits of its quantum output. In general, for input $k$, we do not really need to run a program that computes $\rho$, instead we can run any program that outputs a state ``close enough'' to the output of $\mathsf{U}_{t_{\rho}}(ak)$ for some natural number $a \geq 2$. We want to construct a small set of programs that are candidates to approximate $\sigma$ in this way for a given $k$ and then specify one of them by giving its index in an enumeration of that set. First, consider a slightly relaxed version of the $P_{s}^{\eta}$ family of sets
\begin{equation}
    \label{eq:pseta2}
    _{2}P_{s}^{\eta} = \left\lbrace t \in P_{s}  | \; \textnormal{for all} \; 1 \leq x \leq s, \; D(\eta(t,x), \eta)\leq \frac{2}{x}  \right\rbrace.
\end{equation}
Clearly, $P_{s}^{\eta} \subseteq~ _{2}P_{s}^{\eta}$ for all parameters $s, \eta$, and they are computably enumerable in the same way the $P_{s}^{\eta}$ are. Additionally, note that for any $t \in P_{s}^{\sigma}$ and $1 \leq x \leq s$ it holds that 
\begin{equation}
    D(\eta(t,x), \eta(t_{\rho}, s)) \leq  D(\eta(t,x), \sigma) +  D(\sigma, \eta(t_{\rho}, s)) \leq \frac{1}{x} + \frac{1}{s} \leq \frac{2}{x},
\end{equation}
and thus, for any $s$ we have that $P_{s}^{\sigma} \subseteq ~_{2}P_{s}^{\eta(t_{\rho}, s)}$ and, following Equation~\eqref{eq:boundpssigma}, we get that 
\begin{equation}
    \label{eq:2sigmatrho}
    |_{2}P_{s}^{\eta(t_{\rho},s)}| > 2^{\ell}.
\end{equation}
Define now the sets
\begin{equation}
    \label{eq:sigmas}
    \Sigma_s = \left\lbrace t \in P_s | \; |_{2}P_{s}^{\eta(t,s)}| > 2^{\ell} \right\rbrace. 
\end{equation}
The set $\Sigma_s$ is computably enumerable as shown by Algorithm~\ref{alg:sigmasenumeration} (see Appendix). Because we know that $t_{\rho} \in \Sigma_s$ can be used to approximate $\sigma$ (given the value of $m$), we can think of the $\Sigma_s \subseteq P_s$ as a subset of candidates for programs to compute $\sigma$. Unfortunately, it is not straightforward to find a bound for $|\Sigma_s|$. Instead, define $\Sigma_s^*$ as the set of states that can be outputted by Algorithm~\ref{alg:sigmastar}, which constructs a list of programs in $\Sigma_s$ that are all at least $\frac{4}{s}$-distant from each other. The set $\Sigma_s^*$ is computably enumerable by construction and satisfies the following properties:
\begin{enumerate}
    \item For any distinct $t,t' \in \Sigma_s^*$ it holds that the associated $_{2}P_{s}^{\eta(t,s)}, ~_{2}P_{s}^{\eta(t',s)}$ are disjoint.
    \item There is at least one $\tau \in \Sigma_s^*$ such that $D(\eta(t_{\rho},s), \eta(\tau,s)) \leq \frac{4}{s}$ and therefore
    \begin{equation}
    \label{eq:tausigmastar}
        D(\sigma, \eta(\tau,s)) \leq D(\sigma, \eta(t_{\rho},s)) + D(\eta(t_{\rho},s), \eta(\tau,s)) \leq \frac{1}{s} + \frac{4}{s} = \frac{5}{s}
    \end{equation}
\end{enumerate}

\begin{center}
\begin{algorithm}[H]
\caption{Enumeration of the set $\Sigma_s^*$.}
\label{alg:sigmastar}
\textbf{Parameters:} Universal dcq-TM $\mathsf{U}$, enumeration $\mathcal{E}^{\Sigma}_{s,m,|t_{\rho}|}$ of the set $\Sigma_s$ given $s, m, |t_{\rho}|, \ell$ \;
\textbf{Input:} Integers $i, s, m, |t_{\rho}|, \ell$\; 
$L \leftarrow ()$, $j \leftarrow 1$\;
\While{$|L|<i$}{
    $t \leftarrow \mathcal{E}^{\Sigma}_{s,m,|t_{\rho}|}(j)$ \;
    \If{$D(\eta(t,s),\eta(t',s)) > \frac{4}{s}$ for each $t' \in L$}
    {{Add $t$ to $L$;}}
    $j \leftarrow j+1$\;
}
\Return{$L[i]$}
\end{algorithm}
\end{center}

This means that we can compute $\sigma$ as follows. On input $k$, let $j$ be the index of a program $\tau$ that satisfies Equation~\eqref{eq:tausigmastar} for the set $\Sigma_{5k}^*$. Directly compute the values of $m, |t_{\rho}|$, and $\ell$; then run Algorithm~\ref{alg:sigmastar} on input $(j, 5k, m, |t_{\rho}|, \ell)$ to obtain $\tau$, and finally output the $m$-prefix of $\mathsf{U}^{Q}_{\tau}(5k; \epsilon)$, which by Equation~\eqref{eq:tausigmastar} is at most $\frac{1}{k}$-distant from $\sigma$. The size of the index $j \leq |\Sigma_{5k}^*|$ can be bounded by
\begin{equation}
    \label{eq:boundsigmastar}
    |\Sigma_{5k}^*| < \frac{2^{\K(\rho)+O(1)}}{2^{\ell}},
\end{equation}
by noting that
\begin{equation}
    2^{\ell}|\Sigma_{5k}^*| = \sum_{t \in \Sigma_{5k}^*}\!2^{\ell} < \sum_{t \in \Sigma_{5k}^*} \!\!_{2}P_{5k}^{\eta(t,5k)} \leq |P_s| \leq 2^{\K(\rho)+O(1)},
\end{equation}
where we used the fact that the considered sets $_{2}P_{5k}^{\eta(t,5k)}$ are disjoint. Equation~\eqref{eq:boundsigmastar} in turn gives us a bound for the complexity of $\sigma$
\begin{align}
    \K(\sigma) &\leq \log(|\Sigma_{5k}^*|) + 2\log(n) + 2\log(\K(\rho)) + 2\log(\ell) \nonumber \\
               &< \K(\rho) - \ell + 2\log(n\K(\rho)) + 2\log(\ell), 
\end{align}
substituting $\ell$ from Equation~\eqref{eq:boundpssigma} and cancelling terms
\begin{align}
    \label{eq:largec}
    0 &< 2\log\bigl(\K(\rho) - \K(\sigma) +(c-2)\log(n\K(\rho))\bigr) -(c-4)\log(n\K(\rho))  \nonumber \\
      &\leq 2\log\bigl(n\K(\rho) + c\log(n\K(\rho))\bigr) -(c-4)\log(n\K(\rho)) \nonumber \\
      &\leq 2\log\bigl((c+1)n\K(\rho) \bigr) -(c-4)\log(n\K(\rho)) \nonumber \\
      &\leq 2\log(c+1) -(c-6)\log(n\K(\rho)) \nonumber \\
      &\leq 2\log(c+1) -(c-6) \quad \quad (\textnormal{for } c \geq 6 \textnormal{ and } n>1).
\end{align}
We arrive to the contradiction by noting that the last expression on Equation~\eqref{eq:largec} is negative for large enough $c$.
\end{proof}

By applying Theorem~\ref{thm:yesqchainrule} to a bipartite state $\rho_{AB}$ we get our new version of the chain rule
\begin{align}
    \label{eq:newchainrule}
    \K(\rho_{AB}) &= \K(\rho_{A}) + \K(\rho_{AB}|\mathbold{\rho}_{A}) + O\left(\log(n\K(\rho_{AB})) \right) \nonumber \\
    &= \K(\rho_{B}) + \K(\rho_{AB}|\mathbold{\rho}_{B}) + O\left(\log(n\K(\rho_{AB})) \right),
\end{align}
where the second equality comes from the fact that the state $\rho_{AB}$ and the state obtained by swapping the order of the subsystems $A$ and $B$ are related by a $\log(n)$ overhead, which gets absorbed by the last term in Equation~\eqref{eq:yesqchainrule}. Notably, Equation~\eqref{eq:newchainrule} recovers the form of the classical chain rule Equation~\eqref{eq:chainrulek2} when the two subsystems are uncorrelated, that is 
\begin{equation}
    \label{eq:qweakchainrule}
    \K(\rho_{A}\otimes\rho_{B}) = \K(\rho_{A}) + \K(\rho_{B}|\mathbold{\rho}_{A}) + O\left( \log(n\K(\rho_{A}\otimes\rho_{B})) \right).
\end{equation}
Equipped by this result, we can define a correlation-accounting version of mutual information by quantifying by ``how much'' a state deviates from the classical chain rule, we call this quantity \emph{complexity of correlations}:
\begin{definition} (Complexity of quantum correlations) \\
\label{def:comcorr}
Let $\rho_{AB}$ be a bipartite quantum state. Define the complexity of correlations in the subsystem $A$ respect to $B$ by:
\begin{align}
\label{eq:comcorr}
    C(A : B) &= \K(\rho_{AB}|\mathbold{\rho}_B) - \K(\rho_{A} | \mathbold{\rho}_{B}) \nonumber \\
    &=I^{(1)}(A:B) - I^{(Q)}(A:B),
\end{align}
where
\begin{equation}
    I^{(Q)}(A:B) = \K(\rho_{A}) - \K(\rho_{AB}|\mathbold{\rho}_B).
\end{equation}
\end{definition}
Because $\mathbold{\rho}_{B}$ has the information of the dimension of the space $\mathcal{H}_B$, it is straightforward to see that $C(A : B)$ is non-negative (up to a constant). Note that the quantity $I^{(Q)}(A:B)$ should not be interpreted as an information measure because it can be negative. In fact, it is generically negative, reaching its maximum zero value for product states $\rho_A \otimes \rho_B$ (we denote the quantity by $I^{(Q)}$ to highlight its analogous formal mathematical expression to $I^{(1)}$). From Equations~\eqref{eq:newchainrule} and~\eqref{eq:qweakchainrule} we know that $C$ also satisfies symmetry of information: for an $n$-qubit state $\rho_{AB}$ it holds that 
\begin{align}
    C(A : B) - C(B : A) &= \bigl(\underbrace{I^{(1)}(A:B) - I^{(1)}(B:A)}_{ \underbrace{\in O(\log(n \K(\rho_{A}\otimes \rho_{B})))}_{\in O(\log(n \K(\rho_{AB})))} } \bigr) - \bigl( \underbrace{I^{(Q)}(A:B) - I^{(Q)}(B:A)}_{ \in O(\log(n \K(\rho_{AB}))) } \bigr) \nonumber \\ &= O\big(\log(n \K(\rho_{AB}))\big).
\end{align}

\section{Discussion}
\label{sec:conclusionsk}

In this work, we have laid the foundations of a quantum algorithmic complexity theory based on the dcq-TM model. We started by expanding the definition of the machine to allow it to work with mixed states as inputs/outputs and defined the set of computable quantum states. The choice of our machine model is motivated by two of its main properties: a) the set of machines is naturally discrete and shown to be Turing-complete, and b) there is a single, well defined, halting state for every machine. We mention this properties in contrast to quantum-controlled type of machine models whose description may be parameterized by the set of quantum states associated to a given Hilbert space. If the set of machines is described by a continuous set of states, it can lead to super-Turing capabilities.\footnote{Models of computation that allow for performing fully quantum programs have the problem of outputting  uncomputable quantum states. For instance, consider an uncomputable real number $p\in (0,1)$. According to the Physical Church-Turing postulate, there is no physical system able to sample the Bernoulli distribution with parameter $p$. Otherwise, we could approximate arbitrarily $p$, by the law of large numbers and sampling. However, fully quantum programs would be able to create this distribution by preparing the state $\sqrt{1-p}\ket{0}+\sqrt{p}\ket{1}$ and measuring it on the computational basis.} This issue can be solved by allowing only some discrete subset of states (chosen adequately to output approximations of quantum states within a desired accuracy), as described in~\cite{ber:dam:lap:01}. Another, more difficult to solve issue, is that machines running quantum programs do not always have a single, well defined, halting state. Instead, due to the unitary evolution principle of quantum mechanics, they are generally in a superposition of halting and working states, which makes defining computability difficult. This issue has been openly discussed (see, for instance~\cite{ber:vaz:97, mye:97, lin:98, miy:05}). We make the observation that, in practical scenarios, for adequately modeling real-life quantum computers, we see ourselves as classical users interacting with the quantum machine, which is in line with the dcq-TM model.

Another advantage of the dcq-TM model is that, by having explicitly separated classical and quantum tapes, it lends naturally to comparing two of the approaches mentioned in~\cite{vit:00}. That is, it lets us work with quantum states by inputting them as the physical state of the cells of the quantum tape, or alternatively, we can input their classical representation by encoding the corresponding density matrix as a string in the classical tape. We defined the algorithmic complexity of a state $\rho$ as the size of the description of the shortest computable sequence of ``outputtable" states that approximates $\rho$, and showed that it is machine independent. We also showed that any program that computes $\rho$ can be transformed into one that computes $\mathbold{\rho}$, and vice-versa; which means that the resulting complexities $\K(\rho)$ and $\K(\mathbold{\rho})$ are equivalent. Nevertheless, we found that the conditional complexities $\K(\sigma | \rho)$ and $\K(\sigma | \mathbold{\rho})$ are, in general, not the same. Instead, the classical representations contain more descriptive information than the physical system when used as a resource. This led us to study the behaviour of complexity for quantum state copying. By combining Equations~\eqref{eq:computetoduplicate} and~\eqref{eq:duplicatetocompute} we can obtain the relation
\begin{equation}
    \label{eq:knocloning}
    \K(\rho) = \K(\rho \otimes \rho \: | \: \rho) + O(n),
\end{equation}
which can be seen as a statement of how much of an advantage having a copy of a state gives when computing two copies. For any qubit number $n$ and value $k$ of complexity, there is a finite number of quantum states for which $\K(\rho)\leq k$, and an infinite number of states for which $\K(\rho)> k$. It follows that for any $n \in \nats$ and integer function $k(n) \in O(n)$, there is an infinite number of $n$-qubit states $\rho$ for which $\K(\rho) \gg k(n)$. Therefore, for \emph{almost all} quantum states, having access to a copy of the state gives no significant advantage for computing its duplicate. 

Finally, we considered notions of quantum algorithmic mutual information, and found that two classically equivalent definitions are no longer equivalent for the quantum case. We noted that the property that connects these two quantities in the classical case, namely the chain rule, does not hold for general bipartite quantum states. We interpreted this discrepancy to be caused by the presence of quantum correlations. To compute a quantum multipartite state it is not enough to be able to compute each of its parts individually, but one must also describe the inherent correlations present in the joint state. Taking this into account, we developed a generalization of the chain rule for quantum states in the dcq-TM model. This new property let us to propose a candidate for a measure of the algorithmic complexity of correlations in quantum systems, which is symmetric up to a logarithmic term in the complexity of the joint state times the number of qubits. 

Potential future work directions include the study of the relationship between quantum Kolmogorov complexity and Von Neumann entropy, one way to approach this is through generalization of Brudno's Theorem, which relates the complexity rate with the entropy rate for stationary, ergodic sources. Such analysis has already been done in~\cite{ben:kru:etal:06} for Deutsch's machine model~\cite{deu:85}, which is quantum-controlled. Additionally, further properties of complexity of correlations can be explored and a potential theory of algorithmic complexity of multipartite correlations can be developed. 

\section*{Acknowledgements}
The research on this paper was funded under the FCT project QuantumPrime reference: PTDC/EEI-TEL/8017/2020 (50\%) the QuantaGENOMICS project, through the EU H2020 QuantERA II Programme, Grant Agreement No 101017733, CERN/FIS-PAR/0023/2019 (50\%).
The authors acknowledge Fundação para a Ciência e Tecnologia, Instituto de Telecomunicações Research Unit, ref. UIDB/50008/2020, UIDP/50008/2020 and PEst-OE/EEI/LA0008/2013 and LASIGE Research Unit, ref. UIDB/00408/2020 and ref. UIDP/00408/2020.
NP acknowledges the FCT Estímulo ao Emprego Científico grant no. CEECIND/04594/2017/CP1393/CT000. ML acknowledges the PhD scholarship PD/BD/114334/2016. 

\bibliographystyle{plainnat}
\bibliography{bibfile}


\section{Appendix}

\subsection{Proof of Lemma~\ref{lem:quasiduplicationdifelity}}

\begin{definition}
The Bures angle $\Delta(\rho_1, \rho_2)$ between two quantum states $\rho_1, \rho_2$ is defined as:
\begin{equation}
    \label{def:buresangle}
    \Delta(\rho_1, \rho_2) = \cos^{-1}(\sqrt{F(\rho_1, \rho_2)})
\end{equation}
\end{definition} 
The Bures angle is a metric in the space of density operators acting on a given Hilbert space. Moreover, it is contractive under quantum operations~\cite{nie:chu:00}, that is, for any CPTP transformation $T$ it holds that
\begin{equation}
    \Delta(T(\rho_1), T(\rho_2)) \leq \Delta(\rho_1, \rho_2).  
\end{equation}

\begin{lemma}
Let $T:\mathcal{D}(\mathcal{H})\rightarrow \mathcal{D}(\mathcal{H}^{\otimes 2})$ be a CPTP map. Define the {\em copying error} $\mathtt{CE}_T$ of the transformation $T$ on the state $\rho \in \mathcal{D}(\mathcal{H})$ as

\begin{equation}
    \mathtt{CE}_{T}(\rho) = \sqrt{1-F(\rho^{\otimes 2}, T(\rho))}.
\end{equation}
For any $\rho_1, \rho_2 \in \mathcal{D}(\mathcal{H})$ such that $\mathtt{CE}_{T}(\rho_1) + \mathtt{CE}_{T}(\rho_2) \leq \varepsilon \leq \frac{1}{4}$ it holds that

\begin{equation}
    F(\rho_1, \rho_2) \leq \frac{1}{2} - \sqrt{\frac{1}{4} - \varepsilon} \;\; \; \textnormal{or} \;\;\; F(\rho_1, \rho_2) \geq \frac{1}{2} + \sqrt{\frac{1}{4} - \varepsilon}.
\end{equation}
\end{lemma}

\begin{proof}
From the triangle inequality of the Bures angle we know that
\begin{equation}
    \Delta(\rho^{\otimes 2}_1, \rho^{\otimes 2}_2) \leq \Delta(\rho^{\otimes 2}_1, T(\rho_1)) + \Delta(T(\rho_1),T(\rho_2)) + \Delta(T(\rho_2),\rho^{\otimes 2}_2).
\end{equation}
Rearranging terms and taking the $\sin$ on both sides we get
\begin{equation}
    \label{eq:trianglebures}
    \sin\left(\Delta(\rho^{\otimes 2}_1, T(\rho_1)) + \Delta(T(\rho_2),\rho^{\otimes 2}_2)\right) \geq \sin \left(\Delta(\rho^{\otimes 2}_1, \rho^{\otimes 2}_2) - \Delta(T(\rho_1),T(\rho_2)) \right).
\end{equation}
We can bound the left hand side of Equation~\eqref{eq:trianglebures} by recalling that
\begin{equation}
    \label{eq:boundlefttriangle}
    \sin\left(\Delta(\rho^{\otimes 2}_1, T(\rho_1)) + \Delta(T(\rho_2),\rho^{\otimes 2}_2)\right) \leq \underbrace{\sin\left(\Delta(\rho^{\otimes 2}_1, T(\rho_1))\right)}_{\mathtt{CE}_{T}(\rho_1)} + \underbrace{\sin\left( \Delta(T(\rho_2),\rho^{\otimes 2}_2)\right)}_{\mathtt{CE}_{T}(\rho_2)}.
\end{equation}
Additionally, noting that
\begin{align}
     \sin \left(\Delta(\rho^{\otimes 2}_1, \rho^{\otimes 2}_2) - \Delta(T(\rho_1),T(\rho_2)) \right) = \sin &\left(  \Delta(\rho^{\otimes 2}_1, \rho^{\otimes 2}_2) \right) \cos\left( \Delta(T(\rho_1),T(\rho_2)) \right) \nonumber \\ 
     &- \cos\left(  \Delta(\rho^{\otimes 2}_1, \rho^{\otimes 2}_2) \right) \sin\left( \Delta(T(\rho_2),\rho^{\otimes 2}_2)\right), 
\end{align}
and using the following identities/inequalities (derived from the contractive property of the Bures angle and the squaring property of fidelity) we can bound the right hand side of Equation~(\ref{eq:trianglebures}) and write it in terms of fidelity:
\begin{equation}
    \label{eq:boundrighttriangle}
    \begin{gathered}
        \sin \left(  \Delta(\rho^{\otimes 2}_1, \rho^{\otimes 2}_2) \right) = \sqrt{1-F(\rho^{\otimes 2}_1, \rho^{\otimes 2}_2)} = \sqrt{1-F(\rho_1, \rho_2)^2} \\
        \cos\left( \Delta(T(\rho_1),T(\rho_2)) \right) \geq  \cos\left( \Delta(\rho_1,\rho_2) \right) = \sqrt{F(\rho_1,\rho_2)} \\
        \cos \left( \Delta(\rho^{\otimes 2}_1, \rho^{\otimes 2}_2) \right) = \sqrt{F(\rho^{\otimes 2}_1, \rho^{\otimes 2}_2)} = F(\rho_1, \rho_2) \\
        \sin\left( \Delta(T(\rho_1),T(\rho_2)) \right) \leq \sin\left( \Delta(\rho_1, \rho_2) \right) = \sqrt{1-F(\rho_1, \rho_2)}.
    \end{gathered}
\end{equation}
Substituting the inequalities~(\ref{eq:boundlefttriangle}) and~(\ref{eq:boundrighttriangle}) in the left and right hand sides of Equation~(\ref{eq:trianglebures}), respectively, we get:
\begin{align}
    \mathtt{CE}_{T}(\rho_1) + \mathtt{CE}_{T}(\rho_2) &\geq \sqrt{(F(\rho_1,\rho_2))(1-F(\rho_1, \rho_2)^2)} + F(\rho_1, \rho_2)\sqrt{1-F(\rho_1, \rho_2)} \nonumber \\ 
    &\geq F(\rho_1,\rho_2)(1-F(\rho_1, \rho_2)),
\end{align}
where the last step holds because $F(\rho_1,\rho_2) \in [0,1]$. Finally, assume that $ \mathtt{CE}_{T}(\rho_1) + \mathtt{CE}_{T}(\rho_2) \leq \varepsilon \leq \frac{1}{4}$, then the solutions of the quadratic inequality in the interval $[0,1]$ satisfy
\begin{equation}
    \varepsilon \geq F(\rho_1,\rho_2)(1-F(\rho_1, \rho_2)) \implies F(\rho_1, \rho_2) \leq \frac{1}{2} - \sqrt{\frac{1}{4} - \varepsilon} \;\; \; \textnormal{or} \;\;\; F(\rho_1, \rho_2) \geq \frac{1}{2} + \sqrt{\frac{1}{4} - \varepsilon}.
\end{equation}
\end{proof}


\subsection{Proof of Lemma~\ref{lem:maxquasiduplication}}

Lemma~\ref{lem:maxquasiduplication} is a generalization of the basic property which states that for a finite dimensional Hilbert space $\mathcal{H}$, any set $\lbrace \rho_i \rbrace$ of orthogonal states can have at most $\dim(\mathcal{H})$ elements. The lemma states that this property holds even if the states are not quite orthogonal, but instead have pair-wise fidelities upper bounded by $\delta_0(\mathcal{H})$. To prove this result we will start by proving a restricted version of the main lemma, which is then used to prove the more general statement. First, we restrict the sets to consist only of pure states, then we build upon that version removing the requirement of the states to be pure.

\begin{theorem} (Gram-Schmidt orthogonalization)
    \label{thm:gramschmidt}
    Let $\mathcal{H}$ be a Hilbert space and $\lbrace \ket{\psi_i} \in \mathcal{H} \rbrace_{i=1}^{k}$ a normalized, linearly independent set. For $i=1,2,\ldots,k$ define
    \begin{equation}
        \label{eq:gramschmidtnotnormalized}
        \ket{\tilde{e}_i} = \ket{\psi_i} - \sum_{j<i} \bracket{\tilde{e}_j}{\psi_i}\ket{\tilde{e}_j},
    \end{equation}
    and
    \begin{equation}
        \ket{e_i} = \frac{\ket{\tilde{e}_i}}{|\ket{\tilde{e}_i}|}.
    \end{equation}
    The set $\lbrace \ket{e_i} \rbrace_{i=1}^{k}$ is an orthonormal set.
\end{theorem}
The following properties refer directly to the sets defined in Theorem~\ref{thm:gramschmidt}:
\begin{lemma}
    \label{lem:gramschmidtij}
    Let $\mathcal{H}$ be an $N$-dimensional Hilbert space, $0\leq \delta \leq 1$, and $\lbrace \ket{\psi_i} \in \mathcal{H} \rbrace$ a normalized, linearly independent set such that for $i \neq j$
    \begin{equation}
        | \bracket{\psi_i}{\psi_j} |^2 \leq \delta.
    \end{equation}
    The non-normalized Gram-Schmidt vectors $\lbrace \ket{\tilde{e}_i} \rbrace_{i=1}^{k}$ satisfy for $j < i$
    \begin{equation}
        | \bracket{\psi_i}{\tilde{e}_j} | \leq a_{j-1}\sqrt{\delta},
    \end{equation}
    where
    \begin{equation}
    \label{eq:gramschmidtsequance}
        a_j = 
            \begin{cases}
                1 &\quad\text{if } j=0\\
                a_{j-1}^2 + a_{j-1} &\quad\text{if } j \geq 1 \\
            \end{cases}
    \end{equation}
\end{lemma}
\begin{proof}
We proceed by induction over the index $j$. Note that for $j=1$ and any $i > 1$ we have that 
\begin{equation}
    | \bracket{\psi_i}{\tilde{e}_j} | = | \bracket{\psi_i}{\psi_1} | \leq \sqrt{\delta} = a_0\sqrt{\delta}.
\end{equation}
Assume now that for some $i$ and $j < i-1$ it holds that for all $1 < i' \leq i, 1 \leq j' \leq j$ and $j' < i'$
\begin{equation}
    | \bracket{\psi_{i'}}{\tilde{e}_{j'}} | \leq a_{j'-1}\sqrt{\delta},
\end{equation}
then
\begin{align}
    | \bracket{\psi_i}{\tilde{e}_{j+1}} | &= \big| \bracket{\psi_i}{\psi_{j+1}} - \sum_{k \leq j}  \bracket{\tilde{e}_{k}}{\psi_{j+1}} \bracket{\psi_i}{\tilde{e}_{k}} \big| \nonumber \\
    &\leq \underbrace{\big| \bracket{\psi_i}{\psi_{j+1}}  \big|}_{\leq \sqrt{\delta}} + \sum_{k=1}^{j}  \underbrace{\big|\bracket{\tilde{e}_{k}}{\psi_{j+1}} \big| \big|\bracket{\psi_i}{\tilde{e}_{k}} \big|}_{\leq a_{k-1}^2 \delta \:\: \leq \:\: a_{k-1}^2 \sqrt{\delta}} \nonumber \\
    &\leq \left(1 + \sum_{k=1}^{j} a_{k-1}^2  \right) \sqrt{\delta} \nonumber \\
    &= a_{j}\sqrt{\delta}
\end{align}
\end{proof}

\begin{corolary}
\label{cor:gramschmidtii}
Under the assumptions of Lemma~\ref{lem:gramschmidtij}, with $\delta \leq 2^{-3^N}$
\begin{equation}
    \label{eq:gramschmidbound}
    | \bracket{\psi_{i}}{e_{i}} | \geq \big| 1 - 2^{3^N}\delta \big|
\end{equation}
\end{corolary}
\begin{proof}
From Equation~(\ref{eq:gramschmidtnotnormalized}) we know that
\begin{equation}
    \ket{\psi_i}  = \ket{\tilde{e}_i} + \sum_{j<i} \bracket{\tilde{e}_j}{\psi_i}\ket{\tilde{e}_j}.
\end{equation}
Recalling that the two terms on the right side are mutually orthogonal and that $|\ket{\psi_i}| =1$, it follows that $|\ket{\tilde{e}_i}| \leq 1$ and therefore
\begin{equation}
    | \bracket{\psi_i}{e_j} | \geq | \bracket{\psi_i}{\tilde{e}_j} |.
\end{equation}
On the other hand, we can bound the sequence defined in Equation~(\ref{eq:gramschmidbound}) by the doubly exponential sequence
\begin{equation}
    a_{j+1} = a_{j}^2 + a_{j} \leq 2a_{j}^2 \leq 2^{3^j},
\end{equation}
and hence, by Lemma~\ref{lem:gramschmidtij}
\begin{align}
     | \bracket{\psi_i}{e_i} | \geq | \bracket{\psi_i}{\tilde{e}_i} | &\geq \big| 1 - \sum_{j<i}| \bracket{\tilde{e}_j}{\psi_i} |^2 \big| \nonumber \\
     &\geq \big| 1 - (a_i-1)\delta \big| \nonumber \\
     &\geq \big| 1 - (a_n)\delta \big| \nonumber \\
     &\geq 1 - 2^{3^N}\delta 
\end{align}
\end{proof}

\begin{lemma} (Pure states version)
\label{lem:maxquasiduplicationpure}
Let $\mathcal{H}$ be an $N$-dimensional Hilbert space and $A = \lbrace \ket{\psi_i} \in \mathcal{H} \rbrace$ a normalized set such that for $i \neq j$
\begin{equation}
    \label{eq:maxquasiduplicationpure}
    | \bracket{\psi_i}{\psi_j} |^2 \leq \delta_{0}^{\textnormal{pure}}(N).
\end{equation}
with
\begin{equation}
    \label{eq:deltapurestates}
    \delta_{0}^{\textnormal{pure}}(N) = \left( \frac{1-\sqrt{1-\frac{1}{N}}}{1+2^{\left(\frac{3^N+1}{2}\right)}} \right)^2.
\end{equation}
Then $A$ cannot have more than $N$ elements. 
\end{lemma}
\begin{proof}
We intend to show that for $\delta_{0}^{\textnormal{pure}}(N)$ defined as in Equation~(\ref{eq:deltapurestates}), the set $A$ must be linearly independent. Let $A'=\lbrace \ket{\phi_{i}} \rbrace_{i=1}^{M} \subseteq A$ be a maximal linearly independent subset of $A$ and $A'' = A \smallsetminus A'$. We proceed to prove that any other state belonging to the Hilbert space $\mathcal{H}'= \textnormal{Span}(A) = \textnormal{Span}(A')$ cannot satisfy Equation~(\ref{eq:maxquasiduplicationpure}), and thus $A''$ must be the empty set.

Let $\lbrace \ket{e_{i}} \rbrace_{i=1}^{M}$ be the set obtained from $A'$ through the Gram-Schmidt process. Because the $\ket{e_{i}}$ form an orthonormal basis for $\mathcal{H}'$, it holds that for any $\ket{\tilde{\phi}} \in A'' \subseteq \mathcal{H}'$ there exists $k$ such that
\begin{equation}
    \label{eq:fidelitywithbasis}
    |\bracket{\tilde{\phi}}{e_{k}}|^2 \geq \frac{1}{M} \geq \frac{1}{N}.
\end{equation}
Furthermore, from the triangle inequality of the trace distance
\begin{equation}
    D(\ket{\tilde{\phi}}, \ket{e_{k}}) + D(\ket{e_{k}}, \ket{\phi_{k}}) \geq D(\ket{\tilde{\phi}}, \ket{\phi_{k}}),
\end{equation}
writing the trace distances in terms of inner products we obtain
\begin{equation}
    \sqrt{1-\big| \bracket{\tilde{\phi}}{e_{k}} \big|^2} +   \sqrt{1-\big| \bracket{e_{k}}{\phi_{k}} \big|^2} \geq \sqrt{1-\big| \bracket{\tilde{\phi}}{\phi_{k}} \big|^2}.
\end{equation}
We can now use the inequalities~(\ref{eq:gramschmidbound}) and~(\ref{eq:fidelitywithbasis}) to get
\begin{equation}
    \sqrt{1-\frac{1}{N}} +   \sqrt{1-\left(1 - 2^{3^N}\delta_{0}^{\textnormal{pure}}  \right)^2} \geq \sqrt{1-\big| \bracket{\tilde{\phi}}{\phi_{k}} \big|^2},
\end{equation}
which we can solve for $\big| \bracket{\tilde{\phi}}{\phi_{k}} \big|^2$ as follows
\begin{equation}
    \label{eq:fidelitywithoutside}
    \big| \bracket{\tilde{\phi}}{\phi_{k}} \big|^2 \geq 1 - \left( \sqrt{1-\frac{1}{N}} + \sqrt{1-\left(1 - 2^{3^N}\delta_{0}^{\textnormal{pure}}  \right)^2} \right)^2.
\end{equation}
All that remains now is to show that the right side of Equation~(\ref{eq:fidelitywithoutside}) is strictly greater than  $\delta_{0}^{\textnormal{pure}}$. Indeed, setting $a = 1-\sqrt{1-\frac{1}{N}}$ and $b = 2^{3^{N}}$ we can rewrite Equation~\eqref{eq:deltapurestates} as
\begin{equation}
    \sqrt{\delta_{0}^{\textnormal{pure}}} = a -\sqrt{2b \delta_{0}^{\textnormal{pure}}},
\end{equation}
and we can now use the fact that for all $N \geq 1$ it holds that $0 < \delta_{0}^{\textnormal{pure}}(N) \leq 2^{-3^{N}} < 1$ to conclude that
\begin{align}
    \delta_{0}^{\textnormal{pure}} < \sqrt{\delta_{0}^{\textnormal{pure}}} &= a -\sqrt{2b \delta_{0}^{\textnormal{pure}}} \nonumber\\
    &< a - \sqrt{b \delta_{0}^{\textnormal{pure}} (2- b \delta_{0}^{\textnormal{pure}})} \nonumber \\
    &= 1- \underbrace{\left( (1-a) + \sqrt{b \delta_{0}^{\textnormal{pure}} (2- b \delta_{0}^{\textnormal{pure}})}\right)}_{\leq 1} \nonumber \\
    &\leq 1-\left( (1-a) + \sqrt{b \delta_{0}^{\textnormal{pure}} (2- b \delta_{0}^{\textnormal{pure}})}\right)^2 \nonumber \\
    &= 1 - \left( \sqrt{1-\frac{1}{N}} + \sqrt{1-\left(1 - 2^{3^N}\delta_{0}^{\textnormal{pure}}  \right)^2} \right)^2 \nonumber \\
    &\leq \big| \bracket{\tilde{\phi}}{\phi_{k}} \big|^2.
\end{align}
This means that no state $\ket{\tilde{\phi}} \in \mathcal{H}'$ can satisfy both $\ket{\tilde{\phi}} \notin A'$ and $\ket{\tilde{\phi}} \in A$. Therefore, $A = A'$ is a linearly independent set and as such may have at most $N$ elements.
\end{proof}

\begin{lemma} (General statement)
Let $\mathcal{H}$ be an $N$-dimensional Hilbert space and $A = \lbrace \rho_i \in \mathcal{D}(\mathcal{H})\rbrace$ be a set of density operators such that for $i \neq j$:
\begin{equation}
    \label{eq:maxquasiduplicationmixed}
    F(\rho_i, \rho_j) \leq \delta_0(N),
\end{equation}
with
\begin{equation}
    \delta_0(N) =\left( \frac{\delta_{0}^{\textnormal{pure}}(N)}{N} \right)^2.
\end{equation}
Then $A$ cannot have more than $N$ elements. 
\end{lemma}
\begin{proof}
To prove the result we will proceed as follows: for each $\rho_i \in A$ we will define an associated pure state such that if $A$ satisfies Equation~\eqref{eq:maxquasiduplicationmixed} then the associated set $I$ of pure states satisfies Equation~\eqref{eq:maxquasiduplicationpure}, which in turn by lemma~\ref{lem:maxquasiduplicationpure} implies that $|I|=|A| \leq N$.

Let $M = |A|$. Following the spectral decomposition theorem; for $i =1, \ldots ,M$, let $\lbrace \ket{\psi_{i}^{k}}\rbrace_{k=1}^{N}$ be a set of orthogonal states and $\lbrace \lambda_{i}^{k}\rbrace_{k=1}^{N}$ a set of non-negative real numbers such that
\begin{equation}
    \rho_i = \sum_{k=1}^{N} \lambda_{i}^{k} \ketbra{\psi_{i}^{k}}.
\end{equation}
Furthermore, it is easy to check that for $i,j = 1, \ldots ,M$ the following is an orthonormal set:
\begin{equation}
    \label{eq:tricksuperposition}
    \Bigl\{ \ket{\psi_{j}^{s}(i)} =  \sum_{k=1}^{N} \bracket{\psi_{j}^{s}}{\psi_{i}^{k}} \ket{\psi_{i}^{k}} \Bigr\}_{s=1}^{N}. 
\end{equation}
We intend now to use the Uhlmann's theorem to find a bound for inner products $\bracket{\psi_{j}^{s}}{\psi_{i}^{k}}$ in terms of $\delta_0(N)$. Recall that any state of the form 
\begin{equation}
    \label{eq:purification}
    \sum_{k=1}^{N} \sqrt{\lambda_{i}^{k}} \ket{\psi_{i}^{k}}_{\mathcal{H}}\ket{\phi_{i}^{k}}_{\mathcal{H}'}
\end{equation}
is a purification of $\rho_i$ in $\mathcal{H} \otimes \mathcal{H}'$ for any orthogonal basis $\lbrace \ket{\phi_{i}^{k}} \in \mathcal{H}'\rbrace$, where $\mathcal{H}'$ is an ancilliary Hilbert space isomorphic to $\mathcal{H}$. Now, from Ulhmann's theorem:
\begin{align}
    \delta_0(N) \geq F(\rho_i, \rho_j) &= \max_{\Psi_i, \Psi_j} |\bracket{\Psi_i}{\Psi_j}|^2 \nonumber \\
    &\geq \Big| \sum_{k} \sqrt{\lambda_{i}^{k}} {}_{\mathcal{H}}\bra{\psi_{i}^{k}} {}_{\mathcal{H}'}\bra{\psi_{i}^{k}} \sum_{s} \sqrt{\lambda_{j}^{s}} \ket{\psi_{j}^{s}}_{\mathcal{H}} \ket{\psi_{j}^{s}(i)}_{\mathcal{H}'} \Big|^2 \nonumber \\
    &= \Big| \sum_{k,s} \sqrt{\lambda_{i}^{k} \lambda_{j}^{s}} \bracket{\psi_{i}^{k}}{\psi_{j}^{s}} \underbrace{\bracket{\psi_{i}^{k}}{\psi_{j}^{s}(i)}}_{\bracket{\psi_{j}^{s}}{\psi_{i}^{k}}} \Big|^2 \nonumber \\
    &= \Big| \sum_{k,s} \sqrt{\lambda_{i}^{k} \lambda_{j}^{s}} \big| \bracket{\psi_{i}^{k}}{\psi_{j}^{s}} \big|^2  \Big|^2,
\end{align}
where the maximum is taken over all purifications $\Psi_i, \Psi_j$ of $\rho_i,\rho_j$, respectively. Since all terms in the last summation are positive numbers we can conclude that for any $k,s = 1, \ldots, N$:
\begin{equation}
    \label{eq:fidelityeigenvectors}
    \sqrt{\delta_0(N)} \geq \sqrt{\lambda_{i}^{k} \lambda_{j}^{s}} \big| \bracket{\psi_{i}^{k}}{\psi_{j}^{s}} \big|^2.
\end{equation}
Define now the sets 
\begin{equation}
    I_i = \Bigl\{ \ket{\psi_{i}^{k}} \: | \: \lambda_{i}^{k} \geq \frac{1}{N} \Bigr\}; \;\;\; I = \bigcup_{i=1}^{M}I_i,
\end{equation}
notice that the set $I$ must have \emph{at least} $M$ elements because every density operator acting on an $N$-dimensional space must have at least one eigenvalue greater or equal to $\frac{1}{N}$. Substituting $\ket{\psi_{i}^{k}} \in I_i$ and $\ket{\psi_{j}^{s}} \in I_j$ for $i \neq j$ in Equation~\eqref{eq:fidelityeigenvectors} we obtain
\begin{equation}
    \frac{\delta_{0}^{\textnormal{pure}}(N)}{N} = \sqrt{\delta_0(N)} \geq \sqrt{\lambda_{i}^{k} \lambda_{j}^{s}} \big| \bracket{\psi_{i}^{k}}{\psi_{j}^{s}} \big|^2 \geq \frac{1}{N} \big| \bracket{\psi_{i}^{k}}{\psi_{j}^{s}} \big|^2,
\end{equation}
and thus
\begin{equation}
    \delta_{0}^{\textnormal{pure}}(N) \geq \big| \bracket{\psi_{i}^{k}}{\psi_{j}^{s}} \big|^2.
\end{equation}
This means that the set $I$ satisfies the conditions of lemma~\ref{lem:maxquasiduplicationpure} and hence may have \emph{at most} $N$ elements. We can conclude then
\begin{equation}
    M\leq |I|\leq N.
\end{equation}
\end{proof}

\subsection{Enumeration of the $\Sigma_s$}

In the proof of Theorem~\ref{thm:yesqchainrule} the sets 
\begin{equation}
    \Sigma_s = \lbrace t \in P_s | \; |_{2}P_{s}^{\eta(t,s)}| > 2^{\ell} \rbrace, 
\end{equation}
were defined and claimed to be computably enumerable. Algorithm~\ref{alg:sigmasenumeration} provides a way to enumerate them by calling a computable enumeration $\mathcal{E}_s$ of the set $P_s$, and subroutine $\Upsilon$, which on input $(s,t,i)$:
\begin{itemize}
    \item Outputs the $i$-th element of some fixed computable enumeration $\Upsilon_{s,t}$ of the set $_{2}P_{s}^{\eta(t,s)}$ if $i \leq |_{2}P_{s}^{\eta(t,s)}|$.
    \item Does not halt if $i > |_{2}P_{s}^{\eta(t,s)}|$.
\end{itemize}
Note that a subroutine with the above properties can always be constructed from a program that enumerates a set, since such program can be simulated to produce a certificate of correctness for its output whenever the index is less or equal to the size of the set (by printing the transcript of all its computation), so whenever it halts with an invalid certificate, the simulator can be made to enter an infinite cycle. 

\begin{center}
\begin{algorithm}[H]
\caption{Enumeration of the set $\Sigma_s$ given computable enumerations $\mathcal{E}_s, \Upsilon_{s,t}$ of the respective sets $P_s$ and $_{2}P_{s}^{\eta(t,s)}$, and the values of $s$ and $\ell$.}
\label{alg:sigmasenumeration}
\textbf{Parameters:} Subroutine $\Upsilon$, which in input $(s,t,i)$ outputs the $i$-th element of some fixed computable enumeration $\Upsilon_{s,t}$ of the set $_{2}P_{s}^{\eta(t,s)}$. Computable enumeration $\mathcal{E}_s$ of the set $P_s$. Positive integers $s,\ell$ \; 
\textbf{Input:} Integer $i \geq 1$\; 
$L \leftarrow (), j \leftarrow 1$; \\
\While{$|L| < i$}{
    Set $n$ equal to the smallest natural number such that $j < \frac{(n+1)(n+2)}{2}$\;
    \eIf{$n = 1$}{$m=s-n$\;}{$m = s \mod\frac{n(n+1)}{2}$\;}
    $t' \leftarrow \mathcal{E}_{s}(n-m)$\;
    Simulate the first $m$ steps of $\Upsilon(s,t',2^{l}+1)$\;
    \If{The simulated machine reached its halting state and $t'$ is not in $L$}{
        Add $t'$ to $L$
    }
    $j \leftarrow j+1$\;
}
\Return{$L[i]$}
\end{algorithm}
\end{center}
\end{document}